%% file: main.tex
\documentclass[a4paper]{llncs}

\usepackage{environ}
\newcommand{\repeattheorem}[1]{%
	\begingroup
	\renewcommand{\thetheorem}{\ref{#1}}%
	\expandafter\expandafter\expandafter\theorem
	\csname reptheorem@#1\endcsname
	\endtheorem
	\endgroup
}
\NewEnviron{reptheorem}[1]{%
	\global\expandafter\xdef\csname reptheorem@#1\endcsname{%
		\unexpanded\expandafter{\BODY}%
	}%
	\expandafter\theorem\BODY\unskip\label{#1}\endtheorem
}

\newcommand{\repeatlemma}[1]{%
	\begingroup
	\renewcommand{\thelemma}{\ref{#1}}%
	\expandafter\expandafter\expandafter\lemma
	\csname replemma@#1\endcsname
	\endtheorem
	\endgroup
}
\NewEnviron{replemma}[1]{%
	\global\expandafter\xdef\csname replemma@#1\endcsname{%
		\unexpanded\expandafter{\BODY}%
	}%
	\expandafter\lemma\BODY\unskip\label{#1}\endtheorem
}

\usepackage{amsmath,amsfonts,amssymb}
\usepackage{lineno}
\usepackage[shortlabels]{enumitem}

\usepackage{longtable}
\usepackage{url}

\usepackage[dvipsnames]{xcolor}

\usepackage{xspace}

\usepackage{xargs}
\usepackage{xifthen}
\usepackage{framed}
\usepackage{nicefrac}

\usepackage{lipsum}  
\SetLipsumParListSurrounders{\begingroup\color{gray}}{\endgroup}
\SetLipsumSentenceListSurrounders{\begingroup\color{gray}}{\endgroup}

\usepackage{subcaption}
\usepackage{multirow}  
\usepackage{makecell}  

\usepackage{xpunctuate}
\usepackage[all,british]{foreign} 
\redefnotforeign[ie]{i.e\xperiodafter} 
\redefnotforeign[eg]{e.g\xperiodafter}

\usepackage[algo2e, ruled, vlined, linesnumbered, nofillcomment]{algorithm2e}

\SetKwInput{KwInput}{Input}                
\SetKwInput{KwOutput}{Output}              
\SetKw{Continue}{continue}
\SetKw{Break}{break}

\let\oldnl\nl
\newcommand{\nonl}{\renewcommand{\nl}{\let\nl\oldnl}}

\SetCommentSty{commentstyle}
\DontPrintSemicolon

\usepackage{tikz}
\usepackage{ctable}

\usepackage[full,small]{complexity}

\usepackage[draft,english,silent]{fixme}
\fxsetup{theme=color,layout=marginnote,innerlayout=noinline}

\spnewtheorem{fact}[theorem]{Fact}{\bfseries}{\itshape}

\newcommand*\varrule[1][0.4pt]{\leavevmode\leaders\hrule height#1\hfill\kern0pt}
\renewcommand{\G}{\mathbf G}

\def\adm_#1{ \operatorname{adm}_{#1} }
\def\wcol_#1{ \operatorname{wcol}_{#1} }
\def\scol_#1{ \operatorname{scol}_{#1} }

\newcommand{\orderstrict}{\prec}
\def\Vleft_#1{V_{\orderstrict #1}}

\newcommand{\pp}{\operatorname{pp}}
\newcommand{\Candidates}{\mathsf{Cand}}

\makeatletter  
\newcommand*{\defeq}{\mathrel{\rlap{%
                     \raisebox{0.3ex}{$\m@th\cdot$}}%
                     \raisebox{-0.3ex}{$\m@th\cdot$}}%
                     =}
\makeatother

\newcommand{\Neigh}[2]{%
		\ifthenelse{\equal{#2}{}}{N^{#1}}{N_{#1}^{#2}}%
}

\newcommand{\Main}{\texttt{Main}\xspace}

\def\Dvorak{Dvo\v{r}\'{a}k\xspace}

\newcommand*{\circled}[4]{\tikz[baseline=(char.base)]{
             \node[circle,fill=#3,draw=#4,draw,inner sep=1pt,minimum size=1.2em] (char) {\color{#2} #1};}}

\newcommand{\Mark}[1]{\circled{#1}{gray}{white}{gray}}
\makeatletter
\newcommand{\codelabel}[2]{%
	\protected@write \@auxout {}{\string \newlabel {#1}{{#2}{}}}
	\Mark{#2}
}

\makeatother

\renewenvironmentx{leftbar}[2][1=0.5pt, 2=5pt]{%
  \def\FrameCommand{\vrule width #1 \hspace{#2}}\MakeFramed {\advance\hsize-\width \FrameRestore}}%
{\endMakeFramed}

\newcommand{\length}{\mbox{length}}
\newcommand{\dist}{\mbox{dist}}

\renewcommand{\G}{\mathbb G}

\newcommand{\Vias}{\mathsf{Vias}}

\newcommand{\Pack}{\mathsf{Pack}}

\newcommand{\runtime}{$O(m p^7)$\xspace} 
\newcommand{\memory}{$O(n p^3)$\xspace} 
\newcommand{\corpussize}{217\xspace} 
\newcommand{\numnetworks}{209\xspace} 
\newcommand{\nogrownetworks}{93\xspace} 

\begin{document}

\title{\Large A practical algorithm for $3$-admissibility}
\author{%
Christine Awofeso\orcidID{0009-0000-3550-1727} \and
Patrick Greaves\orcidID{0009-0007-0752-0526} \and
Oded Lachish\orcidID{0000-0001-5406-8121}
Felix Reidl\orcidID{0000-0002-2354-3003}
}
\authorrunning{C. Awofeso et al.}
\institute{Birkbeck, University of London, UK \\
\email{\{cawofe01|pgreav01\}@student.bbk.ac.uk, \{o.lachish|f.reidl\}@bbk.ac.uk}}

\maketitle

\begin{abstract}
	The $3$-admissibility of a graph is a promising measure to identify real-world networks that have an algorithmically favourable structure. 

	We design an algorithm that decides whether the $3$-admissibility of an input graph~$G$ is at most~$p$ in time~\runtime and space~\memory, where $m$ is the number of edges in $G$ and $n$ the number of vertices. To the best of our knowledge, this is the first explicit algorithm to compute the $3$-admissibility.
  
  The linear dependence on the input size in both time and space complexity, coupled
  with an `optimistic' design philosophy for the algorithm itself, makes this algorithm practicable, as we demonstrate with an experimental evaluation on a corpus of \corpussize real-world networks.

  Our experimental results show, surprisingly, that the $3$-admissibility of most real-world networks is not much larger than the $2$-admissibility, despite the fact that the former has better algorithmic properties than the latter. 
\end{abstract}

\section{Introduction}\label{sec:Intro}

Our work here is motivated by efforts to apply algorithms from sparse graph
theory to real-world graph data, in particular algorithms that work
efficiently if certain \emph{sparseness measures} of the input graph are
small. In algorithm theory, specifically from the purview of parametrized algorithms, this approach has been highly successful: by designing algorithms
around sparseness measures like treewidth~\cite{courcelleMonadicSecondorderLogic1990,arnborgEasyProblemsTreedecomposable1991,reedAlgorithmicAspectsTree2003},
maximum degree~\cite{seeseLinearTimeComputable1996}, the size of an excluded minor~\cite{demaineAlgorithmicGraphMinor2005}, or the size of a `shallow' excluded minor~\cite{BndExpII,dvovrak2010deciding}, many hard problems allow the design of approximation or parametrized algorithms with some dependence on these measures.

For real-world applications, many of the algorithmically very useful measures turn out to be too restrictive, that is, the measures will likely be too large for most practical instances. Other graph measures, such as degeneracy, might be bounded in practice but provide only a limited benefit for algorithm design. We therefore aim to identify measures that strike a balance: we would like to find measures that are small on many real-world networks \emph{and} providing an algorithmic benefit. Additionally, we would like to be able to compute such measures efficiently.

A good starting point here is the \emph{degeneracy} measure, which captures the maximum density of all subgraphs. Recall that a graph is $d$-degenerate if its vertices can be ordered in such a way that every vertex~$x$ has at most~$d$ neighbours that are smaller than~$x$. Such orderings cannot exist for \eg graphs that have a high minimum degree or contain large cliques as subgraphs.  In a survey of 206 networks from various domains by Drange \etal~\cite{drangeComplexityDegenerate2023}, it was shown that the degeneracy of most real-world networks is indeed small: it averaged about~23 with a median of~9.

Awofeso \etal identified the $2$-admissibility~\cite{Awofeso25} as a promising measure since it provides more structure than degeneracy and therefore better algorithmic properties (see their paper for a list of 
results). The $2$-admissibility is part of a family of measures called \emph{$r$-admissibility} which we define further below. The family includes degeneracy for~$r=1$, intuitively the larger the value~$r$ the `deeper' into the network structure we look. 
Awofeso \etal designed a practical algorithm to compute the $2$-admissibility and experimentally showed that this measure is still small for many real-world networks, specifically for most networks with degeneracy~$d$, the $2$-admissibility is around~$d^{1.25}$. 

Motivated by theoretical results~\cite{FOBndExp,dvorakFOBndExp13,DvorakApproxMeta2022} which imply that graphs with bounded $3$-admissibility allow stronger algorithmic results than graphs with bounded $2$-admissibility\footnote{It is hard to quantify from these algorithmic meta-results how much more `tractable' problems become, though from works like~\cite{reidlColor2023,penaSeshadhri2025} it is clear that \eg some graphs~$H$ can be counted in linear time
in graphs of bounded~$3$-admissibility, while the same is not possible (modulo a famous conjecture) in graphs of bounded~$2$-admissibility.},
we set out to design a practicable algorithm to compute the $3$-admissibility of real-world networks.

\smallskip
\noindent\textbf{Theoretical contribution} 
We design and implement an algorithm that decides whether the $3$-admissibility of an input graph~$G$ is at most~$p$ in time~\runtime using~\memory space.
This improves on a previous algorithm described by \Dvorak~\cite{dvorakDomset2013} with running time $O(n^{3p+5})$ and also beats the $3$-approximation with running time $O(pn^3)$ described in the same paper when~$p < n^{2/7}$. \Dvorak also provides a theoretical linear-time algorithm for $r$-admissibility in \emph{bounded expansion} classes (which include \eg planar graphs, bounded-degree graphs and classes excluding a minor or topological minor); however, this algorithm relies on a data structure for dynamic first-order model checking~\cite{dvorakFOBndExp13}, which certainly is not practical. Our algorithm runs in linear time as long as the $3$-admissibility is a constant; this includes graph classes where \eg the $4$-admissibility is unbounded. 

For reasons of space, we have relegated most theoretical results and their proofs as well as the detailed pseudocode of our algorithm to the Appendix. Our main result is the following:

\begin{reptheorem}{thmMain}\label{thm:Main}
	There exists an algorithm that, given a graph~$G$ and an integer~$p$, decides whether~$\adm_3(G) \leq p$ in time \runtime and space \memory.
\end{reptheorem}

\noindent
\textbf{Implementation and experiments} \\
We implemented the algorithm in Rust and ran experiments on a dedicated
machine with modest resources (2.60Ghz, 16 GB of RAM) on a corpus
of \corpussize networks\footnote{
  Both the implementation as well as the network corpus can be found
  under \url{https://github.com/chrisateen/three-admissibility}
} used by Awofeso \etal. Our algorithm was able to
compute the $3$-admissibility for all but the largest few networks in this
data set (\numnetworks completed, largest completed network is \texttt{teams} with 935K nodes). For more than half of the networks, the
computation takes less than a second, and for 89\% of these the computation took
less than ten minutes. As such, the program is even practicable on higher-end laptops.

Surprisingly, we find that the $3$-admissibility for many networks (\nogrownetworks) is \emph{equal} to the $2$-admissibility, and for the remaining network, the $3$-admissibility is never larger than twice the $2$-admissibility. We discuss why this is indeed surprising, possible explanations, and exciting potential consequences in Section~\ref{sec:experiments}. Detailed experimental results for all networks can be found in the Appendix.

\section{Preliminaries}\label{sec:prelim}

In this paper, all graphs are simple unless explicitly mentioned otherwise. 
\marginnote{$|G|$, $\|G\|$}
For a graph $G$ we use $V(G)$ and $E(G)$ to refer to its vertex set and edge set,
respectively. We use the shorthands $|G| \defeq |V(G)|$ and $\|G\| \defeq |E(G)|$.
The degree of a vertex $v$ in a graph $G$, denoted $\deg_G(v)$, is the number of neighbours $v$ has in $G$.

\marginnote{$x_1 P x_\ell$, avoids}
For sequences of vertices~$x_1,x_2,\ldots,x_\ell$, in particular paths, we use notation like~$x_1 P x_\ell$,~$x_1 Q$ and~$Rx_\ell$ to denote the  subsequences $P = x_2,\ldots,x_{\ell-1}$, $Q = x_2,\ldots,x_\ell$ and $R = x_1,\ldots,x_{\ell-1}$, respectively.
Note that further on, we sometimes refer to paths as having a
\emph{start-point} and an \emph{end-point}, despite the fact that they are not directed. We do so to simplify the reference to the vertices involved.
A path~$P$ \emph{avoids} a vertex set~$L$ if no innerr vertex of~$P$ is contained in~$L$. Note that we allow both endpoints to be in~$L$.
The length of a path $P$ is the number of edges it has and is denoted by $\length(P)$.
The distance between two vertices in a graph $G$, denoted $\dist_G(u,v)$, is the length of the shortest path in $G$ having $u$ and $v$ as endpoints.

\marginnote{$\G$, ordered graph, $\pi(G)$}
An \emph{ordered graph} is a pair $\G = (G, \leq)$ where $G$ is a graph and $\leq$ is a
total order relation on $V(G)$. 
We write $\leq_\G$ to denote the ordering of $\G$ and extend this notation to derive the relations $<_\G$,
$>_\G$, $\geq_\G$.

\section{Graph admissibility}\label{sec:Admissibility}

\newcommand{\Target}{\mathrm{Target}}

To define \emph{$r$-admissibility} we need the following ideas and notation.
\marginnote{$(r,L)$-admissible path}
	Let $\G = (G,\leq)$, $L\subseteq V(G)$ and $v \in V(G)$. 
	A path~$vPx$ is \emph{$(r,L)$-admissible} if its length $\length(vPx)\leq r$ and it avoids~$L$. 
\marginnote{\\$\Target$}
For every $v\in V(G)$, set $L\subseteq V(G)$ and integer $i > 0$ we let $\Target^i_{L}(v)$ be the set of all vertices in $x \in L$ such that $x$~is reachable from $v$ via an $(i,L)$-admissible path. Note that
$\Target^{i}_{L}(v) \subseteq \Target^{i+1}_{L}(v)$.

\marginnote{$(r,L)$-admissible packing, maximum, maximal}
An \emph{$(r,L)$-admissible packing}
is a collection of paths~$\{vP_ix_i\}_i$ with $v$ referred to as the  \emph{root}~$v$ such that every path~$vP_ix_i$ is~$(r,L)$-admissible, the paths~$P_ix_i$ are pairwise vertex-disjoint, and each endpoint~$x_i \in \Target^r_L(v)$. Note that in particular, all endpoints~$\{x_i\}_i$ are distinct. We call such a packing \emph{maximum} if there is no larger $(r,L)$-packing with the same root and \emph{maximal} if the packing cannot be increased by adding a~$(r,L)$-admissible path from~$v$ to an unused vertex in $\Target^r_L(v)$.
We often treat $(r,L)$-admissible packings as trees rooted at~$v$ and use terms such as `parent', `child' or `leaf'. 

An example of a $3$-admissible packing is depicted in Figure~\ref{fig:admissibility}.
We write $\pp^r_L(v)$ to denote the maximum size of any $(r,L)$-admissible packing rooted at~$v$.

\begin{figure}[!t]
	\centering\includegraphics[width=0.75\textwidth]{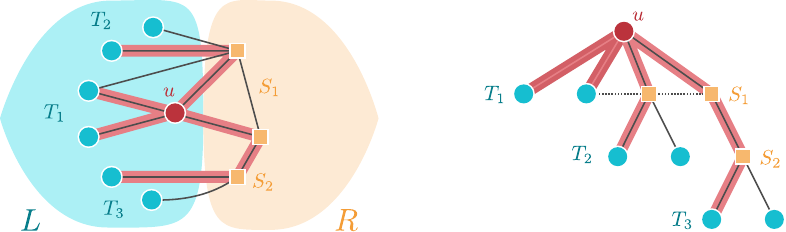}
	\caption{\label{fig:admissibility}%
		On the left, a maximal $(3,L)$-admissible packing rooted at~$u$ as well as the sets
		$T_1 \defeq \Target^1_L(u)$, $T_2 \defeq \Target^2_L(u) \setminus T_1$,
		and~$T_3 \defeq \Target^3_L(u) \setminus (T_1 \cup T_2)$. The sets~$S_i$ contain vertices in~$R$ whose shortest~$(2,L)$-admissible path to~$u$ has length~$i$. On the right, the same local subgraph but embedded in a tree of height~$3$.
	}
\end{figure}

\marginnote{$\pp^r_L$, $\pp^r_\G$}
Given an ordered graph~$\G$, we define $\pp^r_\G(v)$ to be $\pp^r_L(v)$, where $L = \{u \in V(G) \mid u \leq_\G v\}$. 
\marginnote{$\adm_r(\G)$, $\adm_r(G)$}
The \emph{$r$-admissibility} of an ordered graph~$\G$, denoted $\adm_r(\G)$ and the admissibility of an unordered graph~$G$, denoted $\adm_r(G)$ are\looseness-1\footnote{%
Note that some authors choose to define the admissibility as 
	$1 + \max_{v \in \G} \pp^r_\G(v)$ as this matches some other related measures.
} 
\begin{align*}
\adm_r(\G) &\defeq \max_{v \in \G} \pp^r_\G(v) \\
\text{and}~~\adm_r(G) &\defeq \min_{\G \in \pi(G)} \adm_r(\G),
\end{align*}
where $\pi(G)$ is the set of all possible orderings of~$G$.

\noindent
\marginnote{Admissibility ordering}%
If~$\G$ is an ordering of~$G$ such that $\adm_r(\G) = \adm_r(G)$, then we call $\G$ an \emph{admissibility ordering} of~$G$. The $1$-admissibility of a graph coincides with its degeneracy. For~$r \geq 2$, an optimal ordering can also be computed in linear time in $n$ if the class has \emph{bounded expansion}, \ie if the graph class has bounded admissibility for \emph{every}~$r$ (see~\cite{dvorakDomset2013}).

The following fact is a simple consequence of the fact that every $(p,r)$-admissible ordering
is, in particular, a $p$-degeneracy ordering.
\begin{fact} \label{fact:num-edges} 
	If $G$ is $(p,r)$-admissible, then $|E(G)|\le p\cdot|V(G)|$.
\end{fact}

\noindent
The following well-known facts about $r$-admissibility are important for the algorithm we present. These facts hold for all values of~$r$, although we will only need them for~$r \leq 3$. Recall that in a $d$-degenerate graph, we can always find a vertex of degree~$\leq d$. The first lemma shows that a similar property holds in graphs of bounded $r$-admissibility:

\begin{replemma}{lemmaRadm}\label{lemma:Greedy}
	A graph $G$ is $(p,r)$-admissible if and only if, for every nonempty $L
	\subseteq V(G)$, there exists a vertex $u \in V(G)\setminus L$ such that $\pp^r_{L}(u) \leq p$.
\end{replemma}
\begin{proof}
	Suppose first that for every nonempty $L\subseteq V$, there exists a vertex $u \in L$ such that $\pp^r_{V\setminus L}(u) \leq p$.
	Then, an $r$-admissible order of $G$ can be found as follows: first 
	initialise the set $L$ to be equal to $V$ and $i$ to $|G|$, and then repeat the following two steps until $L$ is empty: 
	(1) find a vertex $u \in L$ such that $\pp^r_{V\setminus L}(u) \leq
	 p$ removing it from $L$ and adding it in the $i$'th place of the order
	(2) decrease $i$ by $1$.
	 
	We note that by construction, the $r$-admissibility of the order we got is at most $p$.	Thus, $G$ has $r$-admissibility $p$. 
	
	Suppose that $G$ has $r$-admissibility $p$. Then, there exists an ordering, $\G$ of $V(G)$ such that $\adm_r(\G) \leq p$.
	Let $u_1, u_2, \ldots u_{|G|}$ be the vertices of $\G$ in order. 	
	Let $u_k\in U$ be the vertex with the maximum index in $L$ and define~$U \defeq \{u_1, \ldots, u_k\}$. We note that $L\subseteq U$.
	
	Assume for the sake of contradiction that $\pp^r_L(u_k) > p$. 
	Then, there exists an $(r,L)$-admissible packing $H$ rooted at 
	$u_k$ of size $p+1$. By definition, every path in $H$ starts in $u_k$ and ends in a vertex in~$L$. 
	
	Thus, since $L \subseteq U$, every path $P$ in $H$ either already avoids $U$ or includes a vertex from $U$.	If the second case holds for such a path $P$ then it has the form~$u_kP_1yP_2$ with~$y \in U$
	and~$P_1 \cap U = \emptyset$. Replacing~$P$ by~$u_kP_1y$ and repeating this step for each path that contains a vertex from~$U$ results in a $(r,U)$-admissible packing~$H'$ with~$|H'| = |H|$. Thus~$\pp^r_U(u_k) > p$, contradicting our assumption that~$\G$ has~$\adm_r(\G) \leq p$. 
    \qed
\end{proof}

\noindent
The second lemma allows us to conclude that if during the algorithm run we find a vertex that does not have a large $(3,L)$-path packing for some set~$L$, then we know that this property will hold even if the set~$L$ shrinks in the future:

\begin{replemma}{lemmaCand}\label{lemma:Candidates}
    Let $G$ be a graph, and let $L' \subseteq L \subseteq V(G)$.
    Then for every $v \in V(G)$, we have $\pp^r_{L'}(v) \leq \pp^r_{L}(v)$.
\end{replemma}
\begin{proof}
	Fix a non-empty set~$L \subseteq V(G)$ and an arbitrary vertex~$u \in L$, let~$L' = L \setminus \{u\}$. We show that in this case $\pp^r_{L'}(v) \leq \pp^r_L(v)$ for all~$v \in V(G)$, the claim then follows by induction.

	Assume towards a contradiction that~$\mathcal P'$ is a $(r,L')$-path packing rooted at~$v$ of size~$s > \pp^r_L(v)$. If~$u$ does not appear in~$\mathcal P'$, then $\mathcal P'$ is a $(r,L)$-path packing of size~$s$, a contradiction. The same is true if~$u = v$. Assume therefore that~$P \in \mathcal P'$ contains~$u$ and~$u \neq v$. Let~$Q$ be the segment of~$P$ that starts in~$v$ and ends in~$u$, clearly $|Q| \leq |P| \leq r$ and~$Q$ avoids~$L$. Then $\mathcal P = \mathcal P' \setminus P \cup \{Q\}$ is a
	$(r,L)$-path packing of size~$s$, again contradicting that~$s > \pp^r_L(v)$.
	We conclude that $\pp^r_{L'}(v) \leq \pp^r_L(v)$ and the claim follows.
    \qed
\end{proof}

\noindent
The final lemma is a well-known bound between the sizes of path-packing and the size of the target sets, slightly adapted for our purposes here:

\begin{replemma}{lemmaTargetBound}\label{lemma:TargetBound}
	Fix integers $p, r \geq 1$. Let~$G$ be a $(p,r)$-admissible graph and $L,R$ a partition of $V(G)$ such that for all $u \in R$, $\pp^r_L(u) \leq p$.
	Then for~$u \in R$, $|\Target^r_L(u)| \leq p^r$ and for $v \in L$, $|\Target^r_L(v)| \leq |N_R(v)|(p-1)^{r-1}$.
\end{replemma}
\begin{proof}
	Consider~$u \in R$ first and let~$\Gamma$ be a tree constructed from the shortest $(r,L)$-admissible paths from each vertex in~$|\Target^r_L(u)|$ to~$u$. For every interior vertex~$x \in \Gamma$, we can construct a $(r,L)$-admissible packing by taking one path through each child of~$x$ to a leaf of~$\Gamma$. Since~$\pp_L^r(x) \leq p$, $x$ cannot have more than~$p$ children in~$\Gamma$. The same logic applies to the root, thus~$\Gamma$ has at most~$p^r$ leaves, which therefore bounds~$|\Target^r_L(u)|$.

	We apply the same trick to~$v \in L$, except that for each interior vertex~$x \in \Gamma$, we can also add a path from~$x$ to~$v$ to the packing, hence~$x$ has at most~$p-1$ children. Consequently, $|\Target^r_L(v)|$ contains at most~$|N_R(v)| (p-1)^{r-1}$ vertices.
    \qed
\end{proof}

\begin{figure}[!t]
  \centering\includegraphics[width=0.95\textwidth]{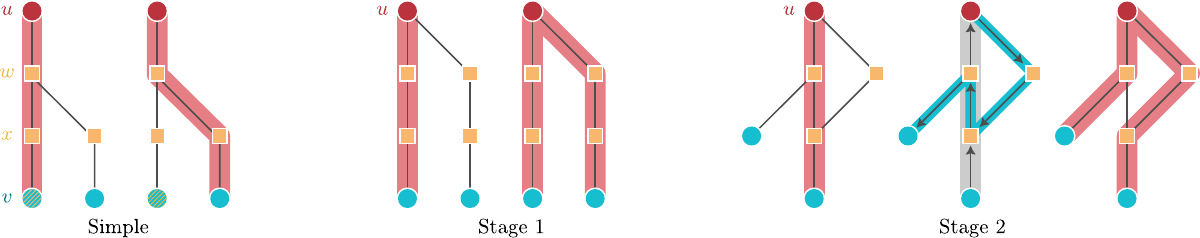}
  \caption{\label{fig:path-packing}%
      The three `escalations' of increasing path packings. During a Simple update (left), the path $uwxv$ is lost since~$v$ moved to~$R$, and the
      algorithm attempts to find a replacement. For small packings, Stage~1 attempts to find disjoint paths rooted at~$u$. If this does not increase the packing size, Stage~2 constructs a suitable flow network to either increase the packing size, or prove that the current packing is maximum.
  }
\end{figure}

\section{Algorithm Overview}
\marginnote{\Main~algorithm}
We provide here a description of how the Algorithm works.
The \Main~algorithm uses an Oracle that iteratively returns the next vertex in the admissibility order (starting at the last vertex and ending with the first). Due to space constraints, we only provide a high-level description of this Oracle here. A formal description, proof of correctness, and analysis of complexity can be found in the appendix.

The input to the \Main~algorithm is a graph~$G$ and an integer~$p$.
We assume that $\|G\| \leq pn$, since this can easily be checked, and if it does not hold, then by Fact~\ref{fact:num-edges}, the $3$-admissibility number of $G$ is strictly greater than $p$. 

Let us for now assume that we have access to an Oracle that, given a subset $L$ of $V(G)$, can provide us with a vertex~$v \in L$
such that~$\pp^3_L(v) \leq p$ if such a vertex exists. With the help of this Oracle, the following greedy algorithm (Algorithm~\ref{alg:Main}) returns an ordering $\G$ such that $\adm_3(\G)\leq p$ if such an order exists and otherwise returns FALSE:
\begin{algorithm2e}[htp]
	\SetAlgoRefName{\texttt{Update}}
	\DontPrintSemicolon
	\KwIn{A graph $G$ and a parameter $p \in [|G|]$}
	Initialise $L \defeq V(G)$,~$R \defeq \emptyset$,~$i \defeq |G|$, and the Oracle.\;
	\While{$L \neq \emptyset$}{ 
		  Ask the Oracle to return a vertex $v \in L$ such that~$\pp^3_L(v) \leq p$.\;
            \If {the Oracle returned FALSE instead of a vertex}{
                return FALSE\;
            }	
            set $v_i \defeq v$\;
            remove $v$ from $L$\;
            prepend $v$ to $R$\;
            set $i \defeq i-1$\;
        }
        Return $R$.
	\caption{\label{alg:Main}%
	Returns a $(p,3)$-admissible ordering of $G$ if one exists and otherwise returns FALSE.}
\end{algorithm2e}

\noindent
This greedy strategy works because of Lemma~\ref{lemma:Greedy}.

Given the Oracle, implementing the above algorithm is straightforward.
The running time of the above algorithm is $O(n)$ plus the overall running time used by the oracle. The same applies for space complexity.
Thus, the problem of finding a $(p,3)$-admissible ordering of a graph $G$ is reduced to the problem of implementing the Oracle, which we outline below after making some structural observations on admissible packings.

\subsection{The structure of $(3,L)$-admissible packings}
\marginnote{covering, chordless}
We need two properties of $(3,L)$-admissible packings that play a central role in the algorithm. Let~$G$ be a graph and $L,R$ a partition of~$V(G)$ and let~$H$ be a $(3,L)$-admissible packing rooted at~$v \in L$. 
Then we call~$H$ \emph{chordless} if for every path~$vwP \in H$ (with $P$ potentially empty) there is no edge between~$v$ and~$P$.

We call~$H$ \emph{covering} if for every vertex $x \in \Target^3_L(v)$ either~$x \in N(v)$ and the path~$xv$ is in $H$, or there exists a $(3,L)$-admissible path from~$x$	to~$v$ which intersects~$V(H)\setminus \{v\}$. In other words, every target vertex of~$v$ is either in the packing or has at least one $(3,L)$-admissible path to~$v$ which intersects~$H$ in a vertex other than~$v$. 

The important observation here is that if for~$v \in L$ there exists a $(3,L)$-packing rooted at~$v$ of size~$k$, then it also has a chordless packing of the same size. If~$vwP \in H$ has a cord, then we can replace~$vwP$ by a shorter path with the same endpoint that is completely contained within~$vwP$.
This observation is already enough to show that one can implement an Oracle for 3-admissibility that works in polynomial time; the following construction provides important intuition and is key in our proof for a much faster Oracle.

\begin{definition}[Packing Flow Network]\label{def:PackingFlowNetwork}
	Let~$L,R$ be a partition of~$V(G)$. For~$u \in L$, the \emph{packing flow network}~$\Pi$ is a directed flow network constructed as follows: start a BFS from~$u$ which stops whenever it encounters a vertex in~$L$ (except~$u$) and stop the process after three steps. Remove all vertices discovered in the third step that are in~$R$. 

	Call the vertices discovered in the $i$th  step~$S_i$ and~$T_i$ ($i \in \{0,\ldots,3\}$), where~$S_i \subseteq R$ to~$T_i \subseteq L$, with~$T_0 = \{u\}$ and~$S_0 = S_3 = \emptyset$. The arcs of~$\Pi$ are the edges of~$G$ from layer~$T_0$ to~$T_1 \cup S_1$, $S_1$ to~$S_2 \cup T_2$, and from~$S_2$ to~$T_2 \cup T_3$.

	The source of the flow network is~$u$ and the sinks are~$T_1 \cup T_2 \cup T_3$. We set the capacities to one for all arcs as well as all vertices, with the exception of~$u$ which has infinite capacity.
\end{definition}

\begin{lemma}\label{lemma:FlowPacking}
	Let~$L,R$ be a partition of~$V(G)$ and let~$\Pi$ be the packing flow network for~$u \in L$. Then there is a one-to-one correspondence between integral flows on~$\Pi$ and chordless $(3,L)$-admissible packings rooted at~$u$.
\end{lemma}
\begin{proof}
	To see that an integral flow corresponds to a packing, note that since the vertices have a capacity of one, every vertex except~$u$ has at most one incoming and one outgoing unit of flow. The saturated arcs therefore form a collection of paths all starting at~$u$ and ending at $T_1 \cup T_2 \cup T_3$. The internal vertices of these paths all lie in~$S_1 \cup S_2$ and since the maximum distance from~$u$ is three, we conclude that all paths are indeed~$(3,L)$-admissible and thus form a $(3,L)$-admissible packing rooted at~$u$.
	Finally, since there are no arcs within~$S_1$ and no arcs from~$u$ to~$T_2 \cup S_2 \cup T_3$, we conclude that the packing is chordless.

	In the other direction, we can convert any chordless $(3,L)$-admissible packing rooted at~$v$ into an integral flow by sending one unit of flow along each path. Since the packing is chordless, all edges are present as arcs in~$\Pi$.
    \qed
\end{proof}

%
\begin{figure}[tbh]
  \hspace*{-7pt}\includegraphics[width=.51\textwidth]{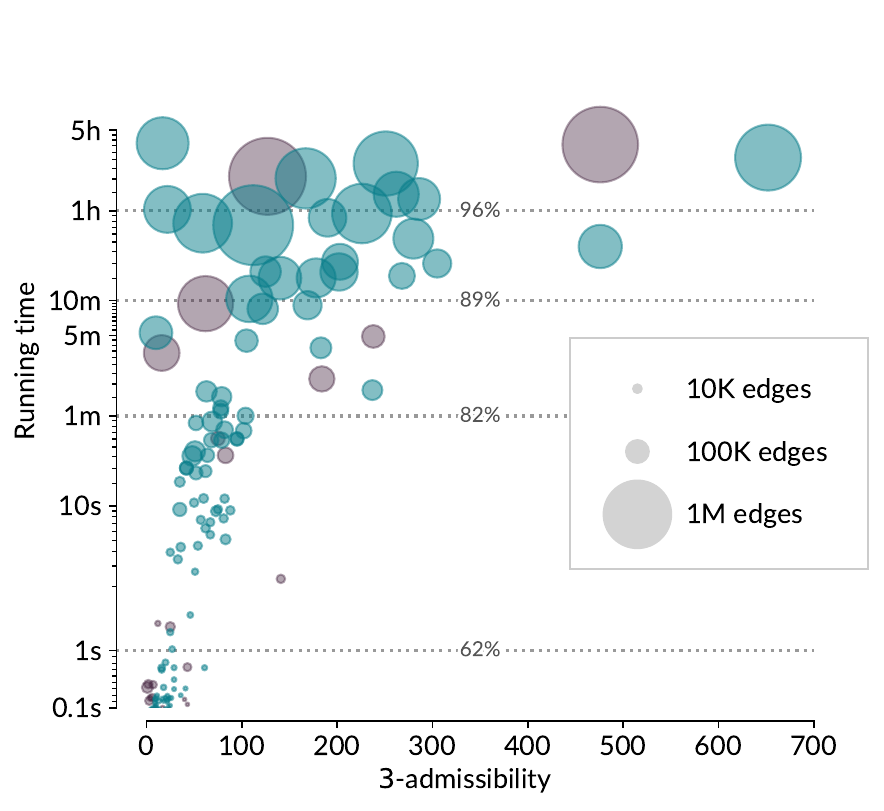}\hspace*{4pt}%
  \includegraphics[width=.51\textwidth]{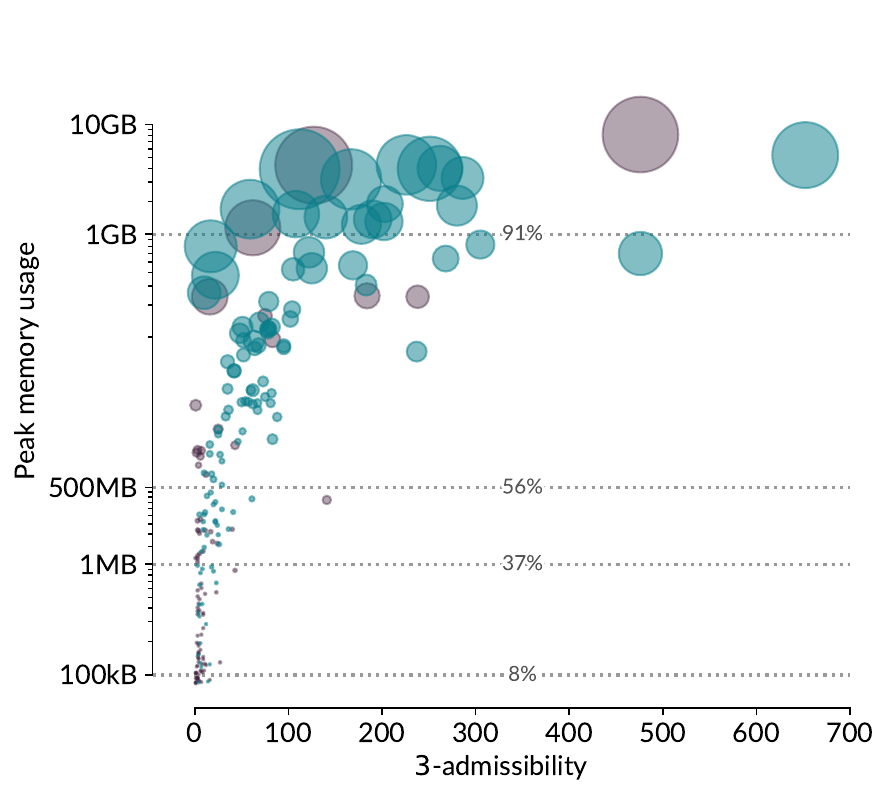}
  \caption{\label{fig:resources}%
    Running time (left) and peak memory consumption (right) of our experiments. Networks where $\adm_2 = \adm_3$ are coloured purple, all other teal. The marker size indicates the number of edges in the network, the horizontal lines show the number of data points below certain interesting thresholds.  }
\end{figure}

\subsection{The Oracle}~\\
\noindent The Oracle has access to the input graph $G$ and the value of $p$.

As already mentioned on every call to the Oracle except for possibly the last, the Oracle returns a vertex.
In the following, the set $L$ is a variable of the oracle that contains all the vertices that the Oracle has yet to return and the set $R$ is a variable that contains all the vertices that the oracle has already returned, \ie $R= V(G)\setminus L$.
We note that the set $L$ has the same role as previously defined and it can be assumed that the oracle maintains this set and not the algorithm that calls it.

When queried, the Oracle must return a vertex $v$ such that $\pp^3_L(v) \leq p$ or FALSE if no such vertex exists.
To do so, the Oracle maintains a set $\Candidates$ and ensures that before every query this set contains exactly the vertices $v\in L$ with $\pp^3_L(v) \leq p$. Hence, the Oracle returns FALSE if~$\Candidates$ is empty and otherwise an arbitrary vertex from this set. In the latter case, the returned vertex is removed from~$\Candidates$ and $L$ and then added to $R$. The Oracle further updates all its data structures to be consistent with the new value of~$L$ and~$R$. 

The challenging part for the Oracle is to ensure that the correct vertices are added to~$\Candidates$.
Note that a vertex will only be removed from $\Candidates$ when it is returned by the Oracle; this is because of the following property: when a vertex $v\in L$ is added to $\Candidates$, we have
$\pp^3_L(v) \leq p$. Since vertices are only removed from the set $L$, Lemma~\ref{lemma:Candidates} ensures that from then on $\pp^3_L(v) \leq p$ will hold in the future.
Next, we explain how the oracle ensures that it adds the correct vertices to $\Candidates$.

In the initialisation stage, before the Oracle receives its first query, $L$ contains all vertices and therefore~$\pp^3_L(v) = \pp^3_{V(G)}(v) = \deg_G(v)$.
Indeed, initially, the Oracle adds to $\Candidates$ all the vertices $v\in L$ such that $\deg_G(v)\leq p$. 
We conclude that the first answer provided by the Oracle is always correct and uses at most $O(n)$ time and space.

The rest of this section is dedicated to explaining if the contents of $\Candidates$ were correct before some call to the Oracle; the Oracle can efficiently ensure that the contents of $\Candidates$ are correct before the next call. Taken together with the observation that the state of the Oracle is initially correct, this implies inductively that the Oracle works correctly.

For every vertex $v \in L\setminus \Candidates$, the oracle maintains a $(3,L)$-admissible packing $\Pack(v)$, which is updated every time a vertex is removed from $L$. For these packings, we maintain two invariants, namely that they are all \emph{covering} and \emph{chordless}. 

Once a vertex is added to $\Candidates$, the oracle stops maintaining its packing until it is removed from $\Candidates$ and added to $R$. 
At this stage, the existing packing is discarded and a new maximal $(2,L)$-packing is computed for the vertex and maintained. The Oracle guarantees that these packings are \emph{chordless}, and the \emph{covering} property is implied by their maximality. As these packings only decrease in size over time and initially contain at most~$p$ paths, operations on these packings are cheap.

Let us now discuss how the Oracle decides when to add a vertex to $\Candidates$.
Every time a vertex $v$ is moved from $L$ to $R$, the oracle updates all the packings of vertices in $R$ and $L\setminus \Candidates$ which include~$v$. This includes trying to add `replacement' paths in case the move of~$v$ invalidated a path; these replacements ensure that the packing invariants are maintained. If for a vertex~$u \in L$ the packing size could not be increased
in this way, and the packing size has reached~$p$, then the Oracle `escalates'
by attempting to add further paths in more computationally expensive ways (see Figure~\ref{fig:path-packing}). First, it attempts to find a path that intersects the current packing only in~$u$. If it finds such a path, it adds it to $\Pack(u)$ which increases its size to $p+1$ and $u$ is not added to $\Candidates$ in this round. If such a path does not exist, then~$\Pack(u)$ is a \emph{maximal} packing. The Oracle then attempts to increase the packing size by constructing a small auxiliary flow graph and augmenting a flow corresponding to the current packing. Here, our theoretical contribution is to show how to construct a small flow graph that mimics the properties of the complete packing flow network (Definition~\ref{def:PackingFlowNetwork}). If the flow increases, the Oracle recovers a $p+1$ packing for~$u$ and does not add~$u$ to $\Candidates$ in this round. If the flow cannot be increased further, the packing $\Pack(u)$ is already maxim\emph{um}---no $(3,L)$-admissible packing rooted at~$u$ of size larger than~$p$ exists; hence $u$ is added to $\Candidates$.

This approach ensures that the Oracle only resorts to performing the costly flow computation if the
current packing is already small. The invariants of chordless and covering $(3,L)$-admissible packings are central here, since it brings the following advantages:
\begin{enumerate}[i.]
	\item Given a covering $(3,L)$-admissible packing rooted at~$v\in L$, when moving $v$ from $L$ to $R$ the vertices in $L$ whose packing need to be updated can be efficiently found with the help of a maximal $(2,L)$-admissible packings stored for vertices in~$R$.
\item Updating a chordless, covering $(3,L)$-admissible packing can be done efficiently when the size of the packing is strictly larger than $p+1$.
\end{enumerate}

\noindent
For the pseudocode of the various parts of the algorithms, proof of correctness, and running time, we refer the reader to the appendix.

\section{Experimental evaluation}\label{sec:experiments}

We implemented the algorithm in Rust (2024 edition, \texttt{rustc} version 1.88.0) and ran experiments on a dedicated machine with 2.60Ghz AMD Ryzen R1600 CPUs and 16 GB of RAM. To optimise performance, we used the compile flag \texttt{target-cpu=native} and
settings \texttt{codegen-units = 1}, \texttt{lto = "fat"}, \texttt{panic = "abort"}, and \texttt{strip = "symbols"} to minimize the final binary size.

Of the \corpussize networks in the corpus, our experiments were completed on \numnetworks.

\subsection*{Computing the $3$-admissibility is practicable}
\noindent
The `optimistic' design of the algorithm, which avoids expensive computations like the flow augmentation as much as possible, resulted in a practical implementation that computed the $3$-admissibility on networks with up to 592K nodes. Figure~\ref{fig:resources} summarises the results: more than half of the networks finished in less than a second, 82\% in under a minute (including a network on 33K nodes), and 89\% in under ten minutes (including a network on 225K nodes).
Memory usage is also modest by modern standards, with more than half of the networks needing less than 500MB and 91\% needing less than 1GB.

This means that our implementation can run even on a modest laptop for quite large networks. As a point of comparison, the computation for the \texttt{offshore} network (278K nodes, 505K edges) ended in about 40 minutes on a laptop with a 2.40GHz Intel i5-1135G7 processor while using about 500MB RAM.

\begin{figure}[tb]
  \centering\includegraphics[width=.8\textwidth]{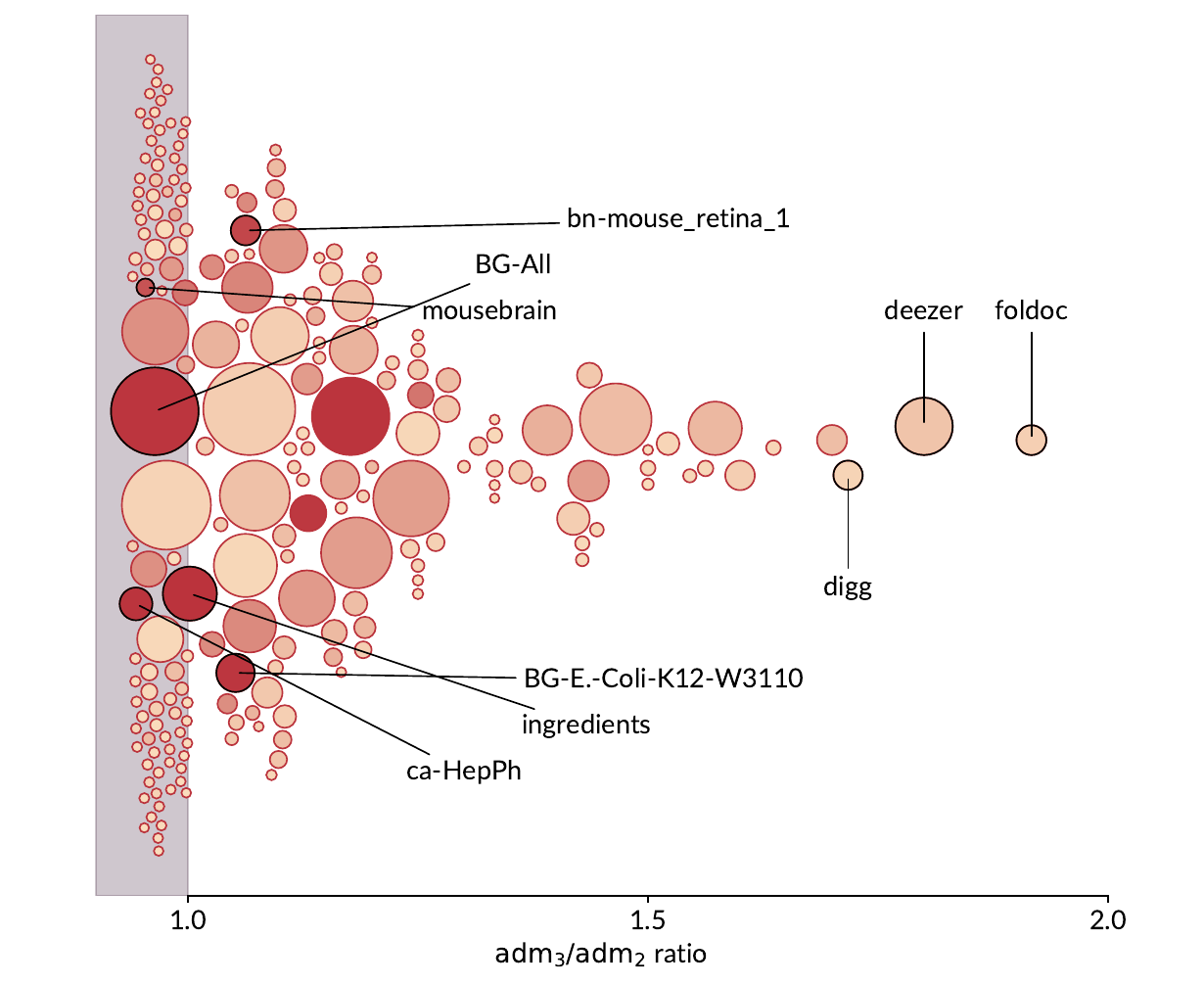}
  \caption{\label{fig:ratio}\small%
    The horizontal position indicates the ratio $\adm_3 / \adm_2$, networks in the grey strip all satisfy~$\adm_3 = \adm_2$ and are distributed horizontally to keep the plot area small. The colour indicates the network degeneracy (as a measure of density), higher values are darker. Networks with the highest ratio are labelled, as well as those networks which have a high density but a low $\adm_3 / \adm_2$ ratio.
  }
\end{figure}

\subsection*{The $3$-admissibility is surprisingly small}
\noindent
To our surprise, the $3$-admissibility for all tested networks is not much larger than the $2$-admissibility. In fact, for 90 of the networks, both values are the same, and for the remaining 110 networks, the $3$-admissibility is less than twice the $2$-admissibility (the largest factor in the data set is 1.91). 
Figure~\ref{fig:ratio} visualizes these ratios, Figure~\ref{fig:resources} shows the absolute values for~$\adm_3$.

We find these results surprising for two reasons. First, experimental measurements of a related measure, the \emph{weak $r$-colouring number} $\wcol_r$, performed by Nadara \etal~\cite{WcolExperimental} showed that $\wcol_3$ was substantially larger than~$\wcol_2$ across most instances. Second, Awofeso \etal~\cite{Awofeso25} showed that the $2$-admissibility is about~$d^{1.25}$, where~$d$ is the degeneracy (the $1$-admissibility) of the network, so we expected to see a similar relation going from~$r = 2$ to~$r = 3$.

There are two plausible explanations for this observation. For some networks, we
could be seeing a plateau at~$r = 3$, \eg; for some~$r \geq 4$, we would see an
increase again. This would be plausible in \eg road networks or other infrastructure
networks which contain longer paths connecting hub vertices; however, then we would expect to see a similar effect in the experiments by Nadara \etal.

The second explanation is that this is indeed the maximum value for \emph
{any}~$r$, which would have quite significant implications about the
structure of such networks: graphs with bounded `$\infty$-admissibility' can be
thought of as gluing together graphs of bounded degree with a constant number
of high-degree vertices added in (see the recent survey by Siebertz~\cite
{siebertzSurvery2025} for a good overview). This is consistent with the
observations by Nadara~\etal, since graphs of bounded degree will have $\wcol_r$ increasing with~$r$, while~$\adm_r$ would be bounded by a universal constant since a path-packing rooted at some vertex~$v$ is always limited by the degree of~$v$.

\section{Conclusion}

We demonstrated that a careful algorithm design and an `optimistic' approach resulted in a resource-efficient implementation to compute the $3$-admissibility of real-world networks.

Our experiments not only demonstrate that the implementation is of practical use, but also that the $3$-admissibility of all \numnetworks networks was surprisingly low; for almost half, it was even equal to the $2$-admissibility. As we outlined above, it is likely that the $r$-admissibility of many real-world networks is already maximal for $r=2$, which has interesting implications for the structure of such networks. 

In the future, we intend on investigating whether the structure theorem of graphs with bounded `$\infty$-admissibility' indeed applies to real-world networks in a meaningful way, as suggested by these experimental results. An important step will be the design of a comparable algorithm to compute the $4$-admissibility, using the lessons learned in this work.

\bibliographystyle{siam}
\bibliography{biblio}

\appendix
\section*{Appendix}

\renewcommand\dblfloatpagefraction{.5}
\renewcommand\dbltopfraction{1}
\renewcommand\floatpagefraction{.90}
\renewcommand\topfraction{1}
\renewcommand\bottomfraction{1}
\renewcommand\textfraction{0}

\renewcommand{\thesection}{A}


\subsection*{The Oracle}
\noindent
In this section, we present all the algorithms used by the Oracle, prove their correctness, and analyse their computational complexity.
We start with the algorithm for initialising the Oracle.
Then we proceed to the algorithm that returns a vertex on a call and updates the Oracle data structures.
This algorithm uses a number of routines that are dealt with afterwards.

\begin{algorithm2e}[htp]
	\SetAlgoRefName{\texttt{Initialise\_Oracle}}
	\DontPrintSemicolon
	\KwIn{A graph $G$ and a parameter $p \in [|G|]$}
	Initialise $L \defeq V(G)$,~$R \defeq \emptyset$\;
        $\Candidates \defeq \{v \in V(G) \mid \deg_G(v)\leq p\}$\;
  Initialise~$\Vias = \emptyset$\;
	\For{$v\in V(G)$}{ 
			$\Pack(v) = \emptyset$\;
      \For{$u\in N(v)$}{
		      Add the path $vu$ to the packing $\Pack(v)$\;
      }
  }
	\caption{\label{alg:Initialise}%
	Initialises the oracle.}
\end{algorithm2e}

\begin{algorithm2e}[htp]
	\SetAlgoRefName{\texttt{Collect\_Targets}}
	\DontPrintSemicolon
	$T \defeq \emptyset$\;
	\For{$w \in (\Pack(v) \cap R) \cup \{v\}$}{
			$T \defeq T \cup N_L(w)$\; \label{ln:markA}
			\If{$w \not \in N(v)$}{
				\Continue\;
			}
			\For{$x \in \Pack(w) \cap R$}{
				$T \defeq T \cup N_L(x)$\; \label{ln:markB}
			}
	}

	\Return T\;
	\caption{\label{alg:CollectTargets}%
	Collects the the targets~$\Target^3_L(v)$ for a vertex~$v$}
\end{algorithm2e}

\begin{algorithm2e}[tp]
	\SetAlgoRefName{\texttt{Oracle\_Query}}
	\DontPrintSemicolon
        \If{$\Candidates = \emptyset$}{
            return FALSE\;
        }
        Choose an arbitrary $v \in \Candidates$\;
        $\Candidates \defeq \Candidates \setminus \{v\}$\;
        $T \defeq \ref{alg:CollectTargets}(v)$\;
        $L \defeq L \setminus \{v\}$, $R \defeq R \cup \{v\}$\;
        \tcp{Update vias}
        $\Vias[v] = \emptyset$\;
        \For{$x \in N_R(v)$}{
        	\For{$y \in N_L(v)$}{
        		\If{$|\Vias[x][y]| \leq 2p$}{
        			$\Vias[x][y] \defeq \Vias[x][y] \cup \{v\}$
        		}
        	}
        	\For{$y \in N_L(x)$}{
						\If{$|\Vias[v][y]| == 2p+1$}{
							\Break\;
						}
        		$\Vias[v][y] \defeq \Vias[v][y] \cup \{x\}$
        	}
        }
        \tcp{Update maximal $(2,L)$-admissible packings}
        \ref{alg:Update2Packings}$(v, T)$\; 

        \tcp{Update $(3,L)$-admissible packings}
        \For{$u \in T\setminus\Candidates$}{
            \ref{alg:SimpleUpdate}$(u,v)$\;
            \If{$|\Pack(u)| == p$}{
            	\ref{alg:Stage1Update}$(u)$\;
            }
            \If{$|\Pack(u)| == p$}{
            	\ref{alg:Stage2Update}$(u)$\;
            }
            \If{$|\Pack(u)| == p$}{
                $\Candidates \defeq \Candidates \cup \{u\}$.
            }
        }
        \Return $v$\:
	\caption{\label{alg:Query}%
	Return a vertex and Update Oracle Data Structures.}
\end{algorithm2e}

\begin{algorithm2e}[ht]
	\SetAlgoRefName{\texttt{Update\_2\_Packings}}
	\DontPrintSemicolon
	\KwIn{Vertex~$v$ which moved from~$L$ to~$R$}

	\tcp{Update other packings}
	\For{$u \in N_R(v)$}{
		Remove~$uv$ from~$\Pack(u)$\;
		\If{$\exists y \in N_L(v) \setminus V(\Pack(u))$}{
			Add~$uvy$ to~$\Pack(u)$\;
		}
	}
  \tcp{Add maximal $(2,L)$-admissible packing for~$v$}
  $\Pack(v) \defeq \emptyset$\;
  \For{$y \in T$}{
  	\If{$y \in N(v)$}{
  		Add $vy$ to~$\Pack(v)$\;
  		\Continue\;
  	}
  	\For{$x \in \Vias[v][y]$}{
  		\If{$x \not \in \Pack(v)$}{
  			Add~$vxy$ to $\Pack(v)$\;
  			\Break\;
  		}	
  	}
  }	

	\nonl\hfil\noindent\fbox{%
    \parbox{7cm}{\footnotesize\raggedright
			\textbf{Important:} The packings~$\Pack(u)$, $u \in R$, contain some paths~$uxv$
			where~$v \not \in L$. Finding the affected packings here is too expensive. 
			Instead, we assume that whenever~$\Pack(u)$ is used, these paths are first identified and removed. Since~$|\Pack(u)| \leq p$, this only costs~$O(p)$ time, the same as accessing all vertices
			stored in~$\Pack(u)$.
    }
  }\hfil%
	\caption{\label{alg:Update2Packings}%
	Maintains the maximal two-packings stored for each vertex in~$R$.}
\end{algorithm2e}

\begin{algorithm2e}[ht]
    \SetAlgoRefName{\texttt{Simple\_Update}}
    \DontPrintSemicolon
    \KwIn{A vertex~$u$ whose data needs to be updated
    and the vertex~$v$ which moved from~$L$ to~$R$ and
    $T \defeq \Target^3_L(v)$.}
    Let $uPv \in \Pack(u)$ be the path with endpoints $u,v$\;
    Remove $uPv$ from $\Pack(u)$\;
    Let $w$ be the neighbour of $u$ in $uPv$\;
    \For{$y\in N_L(w)$}{
        \If{$y\not\in V(\Pack(u))$}{
		  Add the path $uwy$ to $\Pack(u)$\;
            \Return\;
        }
    }
    \For{$x\in N_{\Pack(w)}(w)$}{ \label{ln:simple_break}
        \If{$x \in V(\Pack(u))$}{
            \Continue\;
        }
		  \For{$y\in N_L(x)$}{
            \If{$y \in V(\Pack(u))$}{
            	\Continue\;
            }
            \If{$x\in N(u)$}{
                Add the path $uxy$ to $\Pack(u)$\;
                \Break the loop at line~\ref{ln:simple_break} \;
            }\ElseIf{$w \not \in V(\Pack(u))$} {
                Add the path $uwxy$ to $\Pack(u)$\;
                \Break the loop at line~\ref{ln:simple_break} \;
            }
        }
    }
    \If{$|V(uPv)| < 4$}{
      \Return\;
    }
    \For{$y \in N_L(v)$}{ \label{ln:simple_annoying_case}
        \If{$y \in \Pack(u)$}{
          \Continue\;
        }
        \For{$x \in \Vias[v][u]$}{
            \If{$x \in \Pack(u)$}{
                \Continue\;
            }
            Add the path~$uxvy$ to~$\Pack(u)$\;
            \Return\;
        }
    } 
    \caption{\label{alg:SimpleUpdate}%
    Simple update of a packing.}
\end{algorithm2e}

\subsubsection*{Maintaining $\Vias$ and $\Pack(u)$ for~$u \in R$}
\noindent
The $\Vias$ data structure allows us to find for a pair of vertices $y \in L$,
$u \in R$ up to $2p+1$ vertices from~$N(y) \cap N(u) \cap R$, that is,
vertices in~$R$ that connect~$u$ and~$y$. Let us first outline how this data structure can be implemented in theory to avoid the use of a randomized data structure; in our real implementation, we simply use nested hashmaps.

We can implement the first level of~$\Vias$ as a list of pointers, one for each vertex (\eg by assuming the
vertices are normalised to~$[n]$). $\Vias[u]$ for~$u \in L$ then is a list of
key-pointer pairs, where the keys are~$y \in \Target^2_L(u)$. As we show
below, the number of these entries is~$O(p^2)$, thus by sorting the keys
we have an access time of~$O(\log p)$. These keys are all added when~$u$ is
moved from~$L$ to~$R$, no new keys are added. The pointer associated with
each key leads to a list of vertices with up to~$2p+1$ entries; vertices will
be added to this list but never removed.

Let us now prove that~$\Vias$ indeed functions as intended and the running
time cost of maintaining it.

\begin{replemma}{lemma:Vias}
	Assume~$\Vias$ contained the correct information for the partition $L \cup \{v\}, R$ maintained by the Oracle and that~$v \in \Candidates$ was chosen in \ref{alg:Query}. Then after the update of $\Vias$, it again contains the correct information.

	Over the whole run of \ref{alg:Main}, every entry~$\Vias[x]$, $x \in R$,  contains at most~$p^2$ keys and maintaining $\Vias$ costs in total~$O(p^2 m)$ time and takes~$O(p^3 n)$ space.
\end{replemma}
\begin{proof}
	Since the Oracle has not returned FALSE at any previous point, we conclude that~$N_L(x) \leq \pp^3(x) \leq p$ for all~$x \in R \cup \{v\}$.

 	Since the keys stored in~$\Vias[x]$ are in~$\Target^2_L(x)$, by Lemma~\ref{lemma:TargetBound} we never store more than~$p^2$ vertices, so accessing a specific key~$y$ in~$\Vias[x]$ costs at most~$O(\log p)$ as the keys are sorted. Let us now show that \ref{alg:Main} correctly maintains~$\Vias$.

	Consider any pair of vertices~$y \in L$, $x \in R$ such that $v \in N(y) \cap N(x)$. Then~$x \in N_R(v)$ and~$y \in N_L(v)$, therefore~$v$ is added to~$\Vias[x][y]$ unless it already contains~$2p+1$ vertices. This operation costs~$O(\log p + p) = O(p)$ and is performed for~$N_L(v) \cdot N_R(v) \leq p \deg(v)$ pairs.

	Consider now any pair of vertices~$y \in L$, $x \in R$ such that~$x \in N(y) \cap N(v)$. Then~$x \in N_R(v)$ and~$y \in N_L(x)$, therefore~$x$ is added to~$\Vias[v][y]$, unless it already contains~$2p+1$ vertices.
	This operation costs~$O(p)$ and is performed at most~$2p+1$ times.

	The total cost of all of these operations is
	\[
		\sum_{v \in V(G)} O(p^2 \deg(v) + p(2p+1)) = O(p^2m + p^2 n) = O(p^2m).
	\]
	For each vertex, we store at most~$p^2$ keys, each with a list of up to~$2p+1$ entries, so the maximum space used is~$O(p^3 n)$.
    \qed
\end{proof}

\begin{replemma}{lemma:2Packings}
	Consider a run of~\ref{alg:Query} that returns~$v \in V(G)$.
	After calling \ref{alg:Update2Packings}, for all~$u \in R$,
	$\Pack(u)$ contains a chordless, maximal $(2,L)$-admissible packing rooted at~$u$.
	Maintenance of these packings costs a total of~$O(n p^4 \log p)$ time and $O(np)$ space over the whole run of the algorithm.
\end{replemma}

\begin{proof}
	Let us first prove the statement for~$u \in R \setminus \{v\}$ and deal with~$u = v$ afterwards.
	The first claim is easy to verify: if~$\Pack(u)$ was a chordless, maximal $(2,L \cup \{v\})$-admissible packing, then the only reason its size would drop is it contains a path~$uv$ or~$uxv$. Note that in the latter case, since~$\Pack(u)$ is chordless, $v \not \in N(u)$. But then no $(2,L)$-admissible path from~$u$ can contain~$v$; therefore $uxv$ can simply be removed from~$\Pack(u)$, maintaining maximality (note that we defer this removal to the next time~$\Pack(u)$ is accessed since locating~$u$ from~$v$ is too costly).  

	Consider therefore the case that~$uv \in \Pack(u)$. If~$\Pack(u)$ is not maximal after removing~$uv$, that means that there exists~$(2,L)$-admissible path~$uvy$ with~$y \in L$. Note that, by maximality, $y \not \in N(u)$ as otherwise $uy$ would already be in~$\Pack(u)$. If~$uvy$ exists,
	$u \in N_R(v)$ and~$y \in N_L(v)$ and \ref{alg:Update2Packings} adds this path to~$\Pack(u)$. Since~$uwy$ is chordless, so is the resulting packing.

	Let us now analyse~$\Pack(v)$. \ref{alg:Update2Packings} constructs this packing by checking for each vertex~$y \in T$ whether it can be connected to~$v$ by an~$(1,L)$- or $(2,L)$-admissible path. As the first possibility is checked first, the resulting paths are clearly chordless and since we check all vertices in~$\Target^2_L(v) \subseteq T$, the packing is clearly maximal.

	To bound the running time, note that we iterate over~$u \in N_R(v)$ to modify a packing of size~$O(p)$ and search through $N_L(v) \setminus N_L(u)$. With the usual set data structures, this takes time~$O(p \log p)$ since~$|N_L(v)|, |N_L(u)| \leq p$ for a total of~$O(\deg(v) p \log p)$. Summing over all vertices~$v \in G$, this takes a total time of~$O(m p \log p)$. 
	For the construction of~$\Pack(v)$, we iterate over~$|T| \leq O(p^3)$
	vertices and query up to~$O(p)$ vias, where the query costs~$O(\log p + p)$
	(since we only need to locate the entries for~$\Vias[v][y]$ once to iterate over all~$\leq 2p+1$ entries). Testing whether a vertex is contained in $\Pack(v)$ costs~$O(\log p)$, thus the whole construction takes at most
	$O(p^4 \log p)$ time. Overall, maintaining $(2,L)$-admissible packings
	costs therefore~$O(m p\log p + n p^4 \log p ) = O(n p^4 \log p)$ over the whole run of the algorithm.
	As all these packings contain at most~$p$ paths, the space bound~$O(np)$ follows.
    \qed
\end{proof}

\subsubsection*{Maintaining $\Pack(u)$ for~$u \in L$}

\begin{replemma}{lemma:PackingProperties}
	Let~$L,R$ be a partition of~$V(G)$ and let $H$ be a $(3,L)$-admissible packing rooted at~$u \in L$ which is covering and chordless.

	Let~$y \in \Target^3_L(u) \setminus \Target^1_L(u)$. Then there exists either a path~$uwy$ or~$uwxy$, $w,x \in R$, where~$w \in H$.
\end{replemma}
\begin{proof}
	Let~$uPy$ be a $(3,L)$-admissible path that intersects~$H$ and assume~$uPy$ is the shortest among all paths with this property. Consider first the
	case that~$uPy = uxy$, \eg it has length two. Since it intersects~$H$ in a vertex other than~$u$ and~$y$, it must be~$x$. 

	Now consider the case
	that~$uPy = uwxy$, \eg it has length three. If~$w \in H$, we are done. Otherwise, it must be the case that~$x \in H$. If~$x \in N(u)$ we arrive at a 
	contradiction since~$uxy$ is a shorter $(3,L)$-admissible path that intersects~$H$. Therefore~$\dist_H(u,x)$ must be two, let~$w'$ be the parent
	of~$u$ in~$H$. Then~$uw'xy$ is a path with~$w',x \in R$ and~$w' \in H$, as claimed.
    \qed
\end{proof}

\begin{replemma}{lemma:CollectTargets}
	If $\Pack(u)$ satisfies the two invariants (covering and chordless),
	then calling \ref{alg:CollectTargets} returns the set~$\Target^3_L(u)$
	in time~$O(p^2 |\Pack(u)|)$.
\end{replemma}
\begin{proof}
	Let~$H = \Pack(u)$. Since~$H$ is covering, all vertices $\Target^1_L(u)$ are contained in~$H$ and therefore added in line~\ref{ln:markA}
	when the loop iterates over~$w = v$.

	By \ref{lemma:PackingProperties}, for every~$y \in \Target^3_L(u) \setminus \Target^1_L(u)$ there exists either a path~$uwy$ or~$uwxy$, $w,x \in R$, where~$w \in V(H)$. In the first case, $y \in N_L(w)$ and it is added to~$T$ in line~\ref{ln:markA}.
	In the second case, since~$\Pack(w)$ is a maximal $(2,L)$-admissible packing, there must be a~$x \in \Pack(w) \cap R$ with~$y \in N_L(x)$. Hence, the vertex~$y$ is added to~$T$ in line~\ref{ln:markB}.

	To bound the running time, the outer loop iterates over~$\Pack(u)$, therefore does at most~$O(|\Pack(u)|)$ iterations. The inner loop iterates over~$\Pack(w)$, $w \in R$, which contains at most~$p$ paths. Therefore we access~$N_L$ for
	at most~$O(p|\Pack(u)|)$ vertices that are all in~$R$ and hence have at most~$p$ neighbours in~$L$. We arrive at a running time of~$O(p^2|\Pack(u)|)$.
    \qed
\end{proof}

\begin{replemma}{lemma:SimpleUpdate}
	If $\Pack(u)$ satisfies the two invariants (covering and chordless) before~$v$ was moved from~$L$ to~$R$ and \ref{alg:SimpleUpdate}
	with parameters~$u$ and~$v \neq u$, then~$\Pack(u)$ satisfies the invariants 
	after the call to \ref{alg:SimpleUpdate} with~$v$ now in~$R$.

	This update costs~$O(p^3 \log p)$ time.
\end{replemma}
\begin{proof}
	Let us call the packing~$\Pack(u)$ before the call~$H_1$. Let~$uPv$ be the path from~$u$ to~$v$ in~$H_1$, where~$P$ contains between zero and two vertices. Let initially~$H_2$ be the packing~$H_1$
	with~$uPv$ removed, as constructed by \ref{alg:SimpleUpdate} in the first steps, we will now argue how $H_2$ is modified by the algorithm and argue that in all cases the resulting packing satisfies both invariants. We consider~$L, R$ after the move of~$v$, so $v \not \in L$.

	If~$H_2$ as just constructed happens to be a $(3,L)$-covering packing then \ref{alg:SimpleUpdate} returns~$H_2$ since none of the branches that add a path to packing will execute and the invariants clearly hold.

	Otherwise, let~$y \in \Target^3_L(u)$ so that no $3$-admissible path
	intersects~$H_2$. Note for all~$y \in \Target^1_L(u)$, we have that
	$y \neq v$ and, since~$H_1$ is covering, that the path $uy$ is part of
	$H_1$. Thus assume that $y  \not \in \Target^1_L(u)$, which means that we
	can apply Lemma~\ref{lemma:PackingProperties} to obtain a path~$uwzy$
	or~$uwy$ with~$w \in H_1$. This path must intersect~$uPv$, let us first deal with the case that the intersection is~$v$ and~$uPv$ that contains four vertices, which means that~$v \not \in N(u)$ (since~$H_1$ is chordless) and therefore that the uncovered path must have the
  form~$uwvy$ for some~$w \not \in H_2$ as otherwise it would be already covered. This type of path is added in the loop starting
  at \label{ln:simple_annoying_case}, note that if this path exists
  then a suitable vertex~$w' \in \Vias[v][u]$, $w' \not \in H_1$ will be found since we store up to~$2p+1$ vias for each pair. 

  Assume now that the uncovered path~$uwzy$ ($uwy$) does not intersect~$uPv$ in this way. Then~$w \in P$ in both cases since either
	path must intersect~$uPv$, and since~$uPv$ is chordless we know that $w$ is the neighbour of~$u$ in~$Pv$ ($w = v$ is possible if~$P$ is empty).

	Therefore the vertex~$y$ is either in~$\Target^1_L(w)$ or in~$\Target^2_L(w)$. In the first case, \ref{alg:SimpleUpdate} discovers~$y$ in the first loop and, since we assume that~$y \not \in H_2$, adds~$uwy$ to~$H_2$. Since the loop breaks at this point, we need to argue that both invariants hold. First, we already noted that~$y \not \in \Target^1_L(u)$, thus $y \not \in N(u)$. As~$uwy$ is the only path we added to~$H_2$ the resulting packing has the chordless property. Clearly the covering property holds for~$y$, so consider any other vertex~$y' \in \Target^3_L(u)\setminus \Target^1_L(u)$
	with~$y' \not \in H_2$. By applying~\ref{lemma:PackingProperties} as above, we again find that there must be a $(3,L)$-admissible path~$uwx'y'$ or~$uwy'$. Since~$w \in H_2$, either path would be covered and we conclude that~$H_2$ is indeed covering.

	If the first loop executes without adding any path to~$H_2$, note that every vertex~$y' \in \Target^1_L(w)$ is already contained in~$H_2$, a fact that we will use below.

	Consider now the case that~$y \in \Target^2_L(w)$ and assume towards a contradiction that $y$ is \emph{not} covered by~$H_2$ by the end the algorithm, where now~$H_2$ is the packing constructed by the algorithm.  Since~$\Pack(w)$ is a maximal $(2,L)$-admissible packing, there must exist a vertex~$x \in N_{\Pack(w)}(w)$ such that $y \in \Target^1(x)$. If~$x \in H_2$
	at this point, note that the $(3,L)$-admissible path~$uxy$ is covered by~$H_2$, contradicting our assumption. Thus at this iteration of the loop,
	the algorithm must have found~$y$ in the inner loop. Since~$y \not \in H_2$, the inner loop must have reached the branching statements, and the only reason why no path with leaf was added to~$H_2$, was that~$w \in H_2$. But
	then the $(3,L)$-admissible path~$uwzy$ is clearly covered, contradiction.
	We conclude that after the the second loop is finished, $H_2$ is indeed a covering packing.

	It now only remains to show that~$H_2$ is chordless. Simply note that if the
	second loop adds a path of length two ($uxy$) then there cannot be an edge between~$u$ and~$y$ as otherwise, as observed above, $uy$ would already be a path in~$H_2$. If the second loop adds a path of length three ($uwxy$) then the if-statement ensures that~$x \not \in N(u)$ and by the previous observation also~$y \not \in N(u)$. In both cases the paths are chordless, and hence~$H_2$ is chordless. 

	To bound the running time, note that the most expensive part is the second loop, where the outer loop iterates over $|N_{\Pack(w)}(w)| \leq |\Pack(w)| \leq p^2$ vertices and the inner loop over~$|N_L(x)| \leq p$ vertices, for a total of~$O(p^3 \log p)$ time where the log-factor comes from maintaining suitable set-data structures in~$\Pack(u)$.
    \qed
\end{proof}

\subsubsection*{Stage-1-Update}

\begin{algorithm2e}[htp]
    \SetAlgoRefName{\texttt{Stage\_1\_Update}}
    \DontPrintSemicolon
    \KwIn{A vertex $u$}
		$T_{2,3} \defeq \ref{alg:CollectTargets}(u) \setminus N_L(u)$\;
    \For{$w\in N_R(u) \setminus V(\Pack(u))$}{
    		\For{$y \in T_{2,3} \setminus V(\Pack(u))$}{
    			\If{$y \in N_L(w)$}{
    				Add the path~$uwy$ to $\Pack(u)$\; \label{ln:Stg1_markA}
    				\Return\;
    			}
    			\For{$x \in \Vias[w][y]$}{
    				\If{$x \in \Pack(u)$}{
    					\Continue;
    				}
            \If{$x\in N(u)$}{
	          		Add the path $uxy$ to $\Pack(u)$\; \label{ln:Stg1_markB}
	           } \Else{
								Add the path $uwxy$ to $\Pack(u)$\; \label{ln:Stg1_markC}
	           }
						\Return\;    				
    			}
    		}
    }
    \caption{\label{alg:Stage1Update}%
    Tries to add a disjoint path to the packing~$\Pack(u)$.}
\end{algorithm2e}

\begin{replemma}{lemma:Stage1Update}
	If $\Pack(u)$ satisfies the two invariants (covering and chordless) before~$v$ was moved from~$L$ to~$R$ and \ref{alg:Stage1Update}
	was called for~$u$, then~$\Pack(u)$ satisfies the invariants after the call to \ref{alg:Stage1Update} with~$v$ now in~$R$.

	The cost of this update is $O(\deg_G(u) \cdot p^4)$.
\end{replemma}
\begin{proof}
	Let us call the packing~$\Pack(u)$ before the call~$H$. Note that~$|H| = p$ as otherwise this method would not be called. Moreover, by Lemma~\ref{lemma:SimpleUpdate}, the packing~$H$ at this point satisfies both invariants. We consider~$L, R$ after the move of~$v$, so $v \not \in L$. Recall that by Lemma~\ref{lemma:CollectTargets}
	we are assured that $T_{2,3} = \Target^3_L(u) \setminus \Target^1_L(u)$.

	Note that \ref{alg:Stage1Update} either adds no path to~$H$ or a single path. If no path is added, the packing of course still satisfies both invariants. Further, if a path is added, the covering property is maintained. 

	Let us first argue that a path will be found if it exists. Assume
	there exists a path~$uwy$ or~$uwxy$ where~$y \in \Target^3_L(u) \setminus \Target^1_L(u)$, $w \in N_R(u)$ and~$w,x \not \in V(H)$.
	Clearly the vertex~$w$ will be found in the out loop and~$y$ in the second, and if the path is~$uwy$ then it will be added in Line~\ref{ln:Stg1_markA}). If it is $uwxy$ the the innermost loop will find a suitable~$x' \in \Vias[w][y]$ with~$x' \not \in \Pack(u)$ since it will either locate~$x$, or $\Vias[w][y]$ contains~$2p+1$ vertices.
	Since~$|H| = p$ there are at most~$2p$ vertices in~$V(H)$ that could be contained in~$\Vias[w][y]$ (those at distance one or two from $v$in~$H$) and by the pigeon hole principle a suitable $x' \in \Vias[w][y] \setminus V(H)$ exists. We conclude that a disjoint path~$uwx'y$ will be found.

	For the chordless property, simply note then that if a path~$uxy$ (Line~\ref{ln:Stg1_markA}) or~$uwy$ (Line~\ref{ln:Stg1_markB}) is added to~$H$, then $y \not \in \Target^1_L(u)$ and thus~$y \not \in N(u)$ as otherwise $uy$ would already be a path in~$H$. If a path~$uwxy$ is added (Line~\ref{ln:Stg1_markC}), then by the if-statement we have that~$x \not \in N(u)$ and again $y \not \in N(u)$. We conclude that the added path is chordless in either case, and hence the resulting packing is as well.

	Let us now bound the running time, here we will use that~$|\Pack(u)| = p$
	when this update is performed. The call to \ref{alg:CollectTargets}
	costs~$O(p^2 |\Pack(u)|) = O(p^3)$, the returned set~$T$ has also size at 
	most~$O(p^3)$. The outer loop has at most~$\deg_G(u)$ iterations,
	the inner loop~$|T|$ many. Accessing the correct vias entry then costs~$O(\log p)$ to then iterate over at most~$O(p)$ entries. We can neglect the cost of adding a path to the packing, as it only happens once. Hence the total running time	is~$O(\deg_G(u) p^3 (p + \log p))$
	which is subsumed by the claimed time.
    \qed
\end{proof}

\subsubsection*{Stage-2-Update}

\begin{algorithm2e}[t]
    \SetAlgoRefName{\texttt{Stage\_2\_Update}}
    \DontPrintSemicolon
    \KwIn{Vertex $u$}
		\tcp{Step 1}
		Let~$\hat \Pi$ be the orientation of~$H \defeq \Pack(u)$ with arcs 	pointing away from~$u$\;
		$T \defeq \ref{alg:CollectTargets}(u) \setminus N_L(u)$\;		
		$\hat S_1 \defeq N_H(u) \setminus L$, 
		$\hat S_2 \defeq V(H) \setminus (S_1 \cup L)$\;
		$\hat T_{2,3} \defeq V(H) \cap T$\;

		Add~$E_G(\hat S_1, \hat S_2 \cup \hat T_{2,3})$ as arcs to~$\hat \Pi$\;
		Add~$E_G(\hat S_2, \hat T_{2,3})$ as arcs to~$\hat \Pi$\;

		\tcp{Step 2} 
		\For{$y \in \hat S_2 \cup \hat T_{2,3}$}{
			\For{$x \in N_R(u)$}{
				\If{$xy \in E(G)$}{
					Add~$x$, $ux$, and~$xy$ to~$\hat \Pi$\;
					\Break\;
				}
			}
		}

		\tcp{Step 3} 
		\For{$y \in \hat T_{2,3}$}{
			\For{$x \in \hat S_1$}{
				\For{$w \in \Vias[x][y]$}{
					\If{$w \not \in V(H)$}{
						Add~$w$, $xw$, and~$wy$ to~$\hat \Pi$\;
					}
				}
			}
		}
 
		\tcp{Step 4} 
		\For{$y \in \hat T_{2,3}$}{ \label{line:jump1}
			\For{$w \in N_R(u)$}{ 
				\For{$x \in \Vias[w][y]$}{
					\If{$x \not \in V(H)$}{
						Add~$w,x$ and~$uw$, $wx$, $xy$ to~$\hat \Pi$\;
						\Continue with Line~\ref{line:jump1};
					}
				}
			}
		}

		\tcp{Step 5 and 6}
		\For{$w \in \hat S_1 \cup \hat S_2$}{ \label{line:jump2}
			\For{$y \in N_L(w)$}{
				\If{$y \not \in V(H)$}{
					Add~$y$ and~$wy$ to~$\hat \Pi$\;
					\Continue with Line~\ref{line:jump2}\;
				}
			}
			\If{$w \in \hat S_2$}{
				\Continue\;
			}
			\For{$y \in T \setminus \hat T$}{
				\For{$x \in \Vias[w][y]$}{
					\If{$x \not \in V(H)$}{
						Add~$x$, $wx$, $xy$ to~$\hat \Pi$\;
						\Continue with Line~\ref{line:jump2}\;	
					}
				}
			}
		}

		\tcp{Find augmenting path}
		\label{line:PiDone}
		\If{$\exists$ augmenting path~$P$ for~$H$ in~$\hat \Pi$}{ \label{ln:Stg2_markE}
			$\Pack(v) \defeq H \operatorname{\Delta} P$\;
		}
    \caption{\label{alg:Stage2Update}%
    Tries to increase the size of~$\Pack(u)$ using a flow network.}
\end{algorithm2e}

\begin{replemma}{lemma:MaximalPackingTargets}
	Let~$p$ be an integer and let~$L,R$ be a partition of~$V(G)$ such that~$\pp^3_L(v) \leq p$. Let~$u \in L$ and let~$H$ be a chordless, maximal $(3,L)$-admissible packing rooted at~$u$.
	Then~$|\Target^3_L(v)| \leq |H| (p-1)^2$.
\end{replemma}
\begin{proof}
	Since~$H$ is maximal, note that~$\Target^1_L(u) \subseteq V(H)$.
	It is easy to see that a maximal packing is also covering, thus by Lemma~\ref{lemma:PackingProperties}, for every~$y\in \Target^3_L(u) \setminus \Target^1_L(u)$ there exists either a path~$uwy$ or~$uwxy$ with~$w,x \in R$ and~$w \in H$. We construct a tree~$\Gamma$ by starting with~$\Gamma = H$, and then for each vertex~$y\in \Target^3_L(u) \setminus \Target^1_L(u)$ with~$y \not \in \Gamma$ we add either the path~$uwy$ or~$uwxy$ to~$\Gamma$. At the end of this process, the leaves of~$\Gamma$ are exactly~$\Target^3_L(u)$.

	Note that~$N_\Gamma(u) = N_H(u)$ the first vertex~$w$ on each added path is already in~$V(H)$. Accordingly, $u$ has exactly~$|H|$ children. Each interior vertex~$x \in \Gamma$ lives in~$R$ and note that we can construct a $(2,L)$-admissible packing rooted at~$x$ by routing one path into each subtree of~$x$ and one path to the root~$u$. Since~$\pp^2_L(x) \leq \pp^3_L(x) \leq p$, we conclude that each interior vertex of~$\Gamma$ has at most~$p-1$ children. Thus, $\Gamma$ has at most~$|H| (p-1)^2$ leaves, proving the claim.
    \qed
\end{proof}

\marginnote{Vertex capacity reduction}
\noindent
Recall that in packing flow network, all vertices except the root have unit capacity and all arcs have unit capacities as well. Networks with vertex capacities can be reduced to networks with only edge capacities, to that end each vertex~$v$ is split into two vertices~$v^-$ and~$v^+$ with the arc~$v^-v^+$, with capacity equal to the vertex capacity, between them. Then all arcs~$uv$ from the original network are changed to~$u^+v^-$.

\marginnote{Augmenting path}
We want to avoid this construction in the theoretical analysis, therefore we need the following variation of augmenting paths:
\begin{definition}[Augmenting path]\label{def:AugmentingPath}
	An \emph{augmenting path} in a packing flow network $\Pi$
	with flow~$f$ is a path~$P$ from the source of~$\Pi$ to one of the sinks of~$\Pi$ with the additional property that if~$P$
	enters a saturated vertex via an unsaturated arc, it must leave
	via a saturated arc.
\end{definition}
Paths of this type in a packing flow network are equivalent to augmenting paths in the network derived via the above vertex splitting reduction.

\begin{replemma}{lemma:AugmentingPath}
	Let~$u \in L$ and let~$H$ be a chordless, maximal $(3,L)$-admissible packing rooted	at~$u$ and let~$f_H$ be the corresponding flow in the packing flow
	network~$\Pi$ for~$u$ as per Lemma~\ref{lemma:FlowPacking} with sets
	$S_1,S_2,T_0,T_1,T_2,T_3$. If~$H$ is not
	maximum, then there exists an augmenting path~$uPy$ for~$f_H$ from~$y \in T \defeq T_2 \cup T_3$ to~$u$ with the following properties:
	\begin{enumerate}
		\item $V(P) \cap T$ is a subset of~$V(H) \cap T$,
		\item At most one vertex in $P \cap N_R(u)$ is not contained in~$H$ and 
				  if it exists it is the first vertex in~$P$.
		\item The beginning of~$uPy$ has either the shape~$uwz$ 
					with~$w \in S_1 \setminus V(H)$ and~$z \in (S_2 \cup T_2) \cap V(H)$, or it has
					the shape~$uwxz$ with~$w \in S_1 \setminus V(H)$, $x \in S_2 \setminus V(H)$, and~$z \in V(H) \cap T_3$.
		\item We let~$uPy = uP_1wP_2y$, where~$w$ is the first vertex of~$H$ on~$P$.	
					Then each vertex~$z \in P_2$ with~$z \not \in V(H)$ must be in~$S_2$
					and appears in a subpath~$azb$ in~$wP_2y$ where~$a \in S_1 \cap V(H)$
					and~$b \in (T_2 \cup T_3) \cap V(H)$ or~$b = y$.
	\end{enumerate}
\end{replemma}
\begin{proof}
	By the correspondence established between flows and admissible packings in
	Lemma~\ref{lemma:FlowPacking}, if $H$ is not a maximum then there must
	exist a flow larger than~$|f|$ in~$\Pi$. Since the network contains vertex
	capacities, this means that there must exist an augmenting \emph{walk} in
	the flow network which can visit every vertex with a finite vertex
	capacity up to two times---this can be easily seen by performing the usual
	reduction for vertex capacity networks by replacing each vertex~$v$ with
	two vertices~$v^-$, $v^+$ with an arc~$v^-v^+$ with the vertex capacity on
	it (in our case$1$) and letting all in-arcs of~$v$ go to~$v^-$ and all
	out-arcs go from~$v^+$. By the same construction we can see that the walk might visit multiple vertices in~$T_1 \cup T_2 \cup T_3$, albeit only once each.

	If we compute the \emph{shortest} such walk, however, it is clear that that this is indeed a path: if a vertex~$v \in S_1 \cup S_2$ with capacity~$1$ is visited more than once, then we can shortcut the walk by removing a loop. The resulting path is clearly still an augmenting path, so let us from now on consider a shortest augmenting path~$uPy$. Note that in particular~$u \not \in P$.

	Regarding the possible locations of~$y$, note that since~$H$ is a maximal path packing,	all vertices in~$N_L(v)$ must be contained in~$H$. Thus all arcs from~$u$ to~$T_1$ are saturated by~$f_H$, and therefore no augmenting path can end in~$T_1$, which leaves~$y \in T_2 \cup T_3$.

	\emph{Property 1}:
	To see that $V(P) \cap T$ is a subset of~$V(H) \cap T$, note that if~$P$ contains~$z \in T$, then since vertices in~$T$ have no out-arcs in~$\Pi$,  the two arcs that~$vPy$ uses will be in-arcs, therefore~$f_H$ has one unit of flow going through exactly one of these two arcs, meaning that~$z \in H$.

	\emph{Property 2}:
	We now prove that at most one vertex in $P \cap N_R(u)$ is not contained in~$H$. Let~$z \in P$ such that~$z \in N_R(u)$ and~$z \not \in H$ and let
	$uPy = uP_1zP_2y$. Now simply note that already~$uzP_2y$ is an augmenting path since the arc~$uz$ is, by assumption, not in~$H$ and hence carries no flow in~$f_H$. Thus if~$uPy$ is a shortest augmenting path, the only such vertex~$z$ must come right after~$u$ on the path.

	\emph{Property 3}:
	The first vertex~$w \in P$ is necessarily outside of~$V(H)$, since it must be
	reached via a non-saturated edge. The successor~$x$ of~$w$ on~$P$ is then either in~$T_2 \cup S_2 \cap V(H)$ which gives us~$uwx$ as the start, or we have~$x \in S_2 \setminus V(H)$ in which case the vertex~$z$ after~$x$ must be in~$T_3$. By the maximality assumption on~$H$, $z$ must be in~$H$ as otherwise~$uwxz$ could be added to~$H$. We conclude that~$z \in V(H) \cap T_3$, as claimed.

	\emph{Property 4}:
	Let let~$uPy = uP_1wP_2y$ with~$z_0$ as the first vertex of~$H$ on~$P$. Consider towards a contradiction that there are two successive vertices~$z_1, z_2$
	in~$P_2$ both of which are not in~$V(H) \cup \{y\}$. Let the relevant subsequences of the path be~$az_1z_2b$ where~$a = w$ and~$b = y$ is possible, but~$a = u$ is not by our choice of~$P_2$.
	
	Note that none of the arcs~$az_1$, $z_1z_2$, $z_2b$ are in~$E(H)$ and
	therefore carry no flow in~$f_H$. But then the flow~$f_H$ augmented
	by~$uPy$ would contain the path~$az_1z_2b$, which by Lemma~\ref
	{lemma:FlowPacking} means that~$az_1z_2b$ is a path in the resulting $
	(3,L)$-packing. This is only possible if~$a = u$, which we know is impossible.

	For a triple~$azb$ in~$wP_2y$, note that the arcs~$az$ and~$zb$ carry no flow
	and therefore lead us further away from~$u$. Therefore~$a \in S_2$ is impossible,
	as~$b$ would already have distance four from~$u$. $a = u$ is also impossible since~$u$ is the start of the path, hence it only remains that~$a \in S_1 \cap V(H)$,
	which implies that~$x \in S_2 \setminus V(H)$.
	Since we established above that~$b \in H$, it follows that~$b \in T_3 \cap V(H)$.
    \qed
\end{proof}

\begin{replemma}{lemma:SubflowNetwork}
	Let~$p > 1$ be an integer such that~$\pp^3_L(x) \leq p$ for all~$x \in R$.
	Let~$u \in L$ and let~$H$ be a chordless, maximal $(3,L)$-admissible packing rooted at~$u$ of size~$p$. Let~$\Pi$ be the flow packing network for~$u$ and let~$f_H$ be the flow corresponding to~$H$. 

	There exists a subnetwork~$\hat \Pi$ of~$\Pi$ with~$V(H) \subseteq V(\hat \Pi)$ with at most~$O(p^2)$ vertices and edges, with the property that $f_H$ can be increased	in~$\Pi$ if and only if~$f_H$ can be increased in~$\hat \Pi$.
\end{replemma}
\begin{proof}
	Let~$\hat \Pi$ be be a flow network constructed as follows:
	\begin{enumerate}
		\item Start out with~$\hat \Pi$ as the subnetwork of~$\Pi$ induced by~$V(H)$. 
		\label{step:add_H}
		\item For each~$y \in (S_2 \cup T_2) \cap V(H)$, if in~$\Pi$ there exists a path~$uxy$ with~$x \in S_1 \setminus V(H)$, add~$x$ to~$\hat \Pi$ as well as the arcs~$ux$ and~$xy$. 
		\label{step:add_u_to_level2}
		\item For each~$y \in (T_2 \cup T_3) \cap V(H)$ and~$x \in S_1 \cap V(H)$, if in~$\Pi$ there exists a path~$xwy$ with~$w \in S_2 \setminus V(H)$, add~$w$ to~$\hat \Pi$ as well as the arcs~$xw$ and~$wy$. 
		\label{step:add_level1_to_level23}		
		\item For each~$y \in (T_2 \cup T_3) \cap V(H)$, if in~$\Pi$ there exists a path~$uwxy$ with~$w \in S_1 \setminus V(H)$, $x \in S_2 \setminus V(H)$, then add
		$w$ and~$x$ to~$\hat \Pi$ as well as the arcs~$uw$,	$wx$, and~$xy$.
		\label{step:add_u_to_level3}
		\item For every $w \in (S_1 \cup S_2) \cap V(H)$, if there exists~$y \in (T_2 \cup T_3) \setminus V(H)$ such that~$wy \in \Pi$, then add~$y$ and the arc~$wy$ to~$\hat \Pi$. We add at most one such vertex~$y$ and arc~$wy$ for each~$w$.
		\label{step:add_direct_Ts}
		\item For every $w \in S_1 \cap V(H)$ for which the previous step has 
		\emph{not} added a neighbour in~$(T_2 \cup T_3)\setminus V(H)$, if there exists	a path~$wxy$ in~$\Pi$ with~$x \in S_2 \setminus V(H)$ and~$y \in (T_2 \cup T_3)\setminus V(H)$, add~$x$ to~$\hat \Pi$ as well as the arcs~$wx$ and~$xy$.
		\label{step:add_indirect_Ts}
	\end{enumerate}
	The capacities of the arcs and vertices are exactly as in~$\Pi$, the source is~$u$ and the sinks are~$V(\hat \Pi) \cap (T_2 \cup T_3)$. 

	Since~$V(H) \subseteq V(\hat \Pi)$, the flow~$f_H$ is well-defined on~$\hat \Pi$ if we ignore the flow into~$T_1$ (since $H$ is maximal, all arcs from~$v$ to~$T_1$ are saturated and no augmenting path will change that, which is why we can ignore this part of the packing/flow).

	Recall that an augmenting path in our definition 
	(Definition~\ref	{def:AugmentingPath}) has the additional restriction that if it enters a saturated vertex via an unsaturated arc, it must exit via a saturated arc. This means that when we take an augmenting path~$PaQbR$ and replace the subpath~$Q$ by a subpath~$Q'$, then if the boundaries vertices~$a,b$ are not saturated and~$P$, $Q'$, and~$R$ are disjoint, then the resulting path will also be augmenting.

	Let now~$uPy$ be an augmenting path for~$f_H$ on~$\Pi$ with the properties promised by Lemma~\ref{lemma:AugmentingPath} and assume that it is the shortest such path between~$u$ and~$T_2 \cup T_3$ in~$\Pi$. Let further~$H'$ be the packing	resulting from augmenting~$f_H$ by~$uPy$ (as per Lemma~\ref{lemma:FlowPacking}, the augmented flow will correspond to a packing). We argue argue that then an augmenting path~$uP'y$ for~$f_H$ on~$\hat \Pi$ exists by constructing it from~$uPy$. We will keep the vertices~$uP \cap V(H)$ and argue that all other vertices can be replaced by suitable alternatives.

	First, consider the case that~$y \not \in \hat \Pi$ and let~$z \in uP$ be the last vertex in~$V(H)$. By Property~4 of
	Lemma~\ref{lemma:AugmentingPath} we then either have that~$uPy$ ends in
	$sy$ or in~$sxy$ with~$x \not \in V(H)$. Since~$sx$ and~$xy$ carry no flow
	in~$f_H$, a path in~$H'$ must end in~$sxy$, but then this path can
	only be~$usxy$, hence~$s \in S_1 \cap V(H)$ and~$y \in T_3$. By construction steps~\ref{step:add_direct_Ts} and~\ref{step:add_indirect_Ts}, therefore~$\hat \Pi$ either contains 
	an arc~$sy'$ with~$y' \in T_2 \setminus V(H)$ or a path~$sx'y'$ with~$x' \in S_2 \setminus V(H)$ and~$y' \in T_3 \setminus V(H)$. In either case,
	we can replace the end of the path~$uPy = uP'sQy$ to obtain~$uP's\hat Q\hat y$, where~$s\hat Q \hat y$ is a path in~$\hat \Pi$, and we claim that this
	is indeed an augmenting path. First, assume that there exists some vertex
	$x \in P' \cap \hat Q$ which makes this not a path.	Since~$x \in \hat Q$, it follows that~$x \not \in V(H)$. But	then the path obtain by going from~$u$ to~$x$ via~$P'$ and then directly to~$\hat y$ via~$\hat Q$ is shorter than~$uPy$, contradicting our assumption of~$uPy$ being a shortest augmenting path. Thus~$uP's\hat Q\hat y$ is indeed a path and since all
	vertices we replaced carry no flow in~$f_H$, it is easy to see that this is still an augmenting path.

	Let us now take care of the beginning of the path. By Property~3 of
	Lemma~\ref{lemma:AugmentingPath}, the path~$uP's\hat Q\hat y$ has the
	shape~$uRzP''s\hat Q \hat y$, where~$R$ contains either one or two vertices
	not contained in~$V(H)$ and~$z \in V(H)$. A path~$u\hat Rz$ 
	with~$\hat R \cap V(H) = \emptyset$	was then added to~$\hat \Pi$
	either in step~\ref{step:add_u_to_level2} or~\ref{step:add_u_to_level3} of the construction, we claim that~$u\hat RzP''s\hat Q \hat y$ is still an augmenting path. By the same argument as above, if we had a joint vertex in~$\hat R \cap P''$ or~$\hat R \cap \hat Q$, this vertex would be outside of~$V(H)$ and therefore carries no flow in~$f_H$, therefore we construct a shorter augmenting path contradicting our assumption about~$uPy$ being shortest.

	Finally, consider the middle part~$P''$ of~$u\hat RzP''s\hat Q\hat y$. Consider a vertex~$z' \in P'' \setminus V(H)$ with~$z' \not \in \hat \Pi$.
	By Property~4 of Lemma~\ref{lemma:AugmentingPath}, $z'$ appears in a subpath
	$az'b$ in~$P''$ with~$a \in S_1 \cap V(H)$ and~$b \in (T_2 \cup T_3) \cap V(H)$. Thus, in step~\ref{step:add_level1_to_level23} of the construction,
	a some vertex~$\hat z \in S_2 \setminus V(H)$ was added to that~$a\hat zb$
	is a path in~$\hat \Pi$. We replace~$z''$ by~$\hat z$ and iterate this process with the remainder of~$P''$ until we arrive at a sequence~$\hat P$
	where~$\hat P \subseteq \hat \Pi$. If~$\hat P$ is not a path, then some
	vertex~$\hat z$ was added twice, but since~$\hat z \not \in V(H)$, it is not saturated by~$f_H$ and by the same short-cutting argument as above, this contradicts our assumption that~$uPy$ is a shortest augmenting path. We conclude that~$\hat P$ is a path and $u\hat Rz\hat Ps \hat Q\hat y$ is finally the claimed augmenting path for~$f_H$ contained in~$\hat \Pi$.

 	In the other direction, simply note that $\hat \Pi$ is a subnetwork of~$\Pi$, therefore if~$f_H$ has an augmenting path in~$\hat \Pi$ that same path is also augmenting for~$f_H$ in~$\Pi$.

 	To bound the size of~$V(\hat \Pi)$, simply note that every construction step adds between one and two vertices to either a vertex of~$V(H)$ or between a pair of vertices in~$V(H)$ and adds at most three arcs for each such addition. Since~$|V(H)| \leq 3p+1$, it follows that~$|V(\hat \Pi)| = O(p^2)$. The subnetwork of~$\Pi$ induced by~$V(H)$ contains at most$p |V(H)|$ edges (Fact~\ref{fact:num-edges}), therefore, the total number of edges is~$O(p^2)$ as well.
    \qed
\end{proof}



\begin{replemma}{lemma:Stage2Update} 
	After the call to \ref{alg:Stage2Update}, $\Pack(u)$ is either a chordless, covering $(3,L)$-admissible packing of size $p+1$, or it is a chordless, maximum $(3,L)$-admissible packing of size~$p$.

	A call to \ref{alg:Stage2Update} costs time~$O(p^5 + p^3 \deg(u))$
	and space~$O(p^3)$.
\end{replemma}
\begin{proof}
	Let~$H = \Pack(u)$.
	It is easy to verify that the graph~$\hat \Pi$ constructed 
	by \ref{alg:Stage2Update} is, after the completion of the line~\ref{line:PiDone}, a version of the subnetwork as described in Lemma~\ref{lemma:SubflowNetwork}. As such, if the packing network~$\Pi$ for~$u$ contains an augmenting path for the flow~$f_H$, then so does~$\hat \Pi$

	Therefore, if \ref{alg:Stage2Update} does not find an augmenting path in Line~\ref{ln:Stg2_markE}, we conclude that~$H$ is a maximum $(3,L)$-admissible packing for~$u$. In this case, $\Pack(u)$ remains unchanged and thus~$|\Pack(u)| = p$ when the algorithm terminates.

	Otherwise, the augmentation operation on~$H$ corresponds to taking the symmetric difference between~$E(H)$ and~$E(P)$ and reassembling the resulting path packing. In this case, $|\Pack(u)| = p + 1$ and we have to show that this packing~$H'$ is covering and chordless. 

	Let~$S_1, S_2$ and~$T_2, T_3$ be the vertex sets as defined
	for~$\Pi$. Let us first show that
	$S_1 \cap V(H) = S_1 \cap V(H')$ and~$(T_2 \cup T_3) \cap V(H) = (T_2 \cup T_3 \cap V(H'))$. Note that a vertex~$x \in V(H)$ is removed from~$H$ through~$uPy$ if the path enters~$x$ and
	exists through saturated arcs, \eg precisely those edges that are incident to~$x$ in~$H$. This cannot happen for~$x \in S_1 \cap V(H)$, since	one of the saturated arcs is~$ux$ but~$ux$ cannot be part of the augmenting path---the first arc on~$uPy$ must be unsaturated. The vertices in~$(T_2 \cup T_3) \cap V(H)$ have as incoming arcs only unsaturated edges, so the~$uPy$ cannot remove them from the packing. 

	Therefore, $N_{H'}(v) \supseteq N_H(v)$, and thus any $
	(3,L)$-admissible path~$uwy$ or~$uwxy$ (\cf Lemma~\ref
	{lemma:PackingProperties}) for~$y \in T_2 \cup T_3$ that intersects~$H$
	in~$w$ will also intersect~$H'$ in~$w$, proving that~$H'$ is indeed covering.

	Lemma~\ref{lemma:FlowPacking} established that every flow in~$\Pi$ corresponds to a chordless packing. Since~$f_{H'}$ is also a flow in~$\Pi$,
	we conclude that~$H'$ is chordless.
 
	Let us now bound the running time of \ref{alg:Stage2Update}. The call
	to \ref{alg:CollectTargets} costs~$O(|H| p^2) = O(p^3)$ and the sets
	$T$ has size at most~$p^3$. The initial graph has size~$O(p)$, adding the arcs for Step~\ref{step:add_H} of the construction therefore cost~$O(p^2)$
	if we query each pair of vertices. Step~\ref{step:add_u_to_level2} cost~$O(p \deg_G(u))$, Step~\ref{step:add_level1_to_level23}~$O(p^2 (p+\log p))$,
	Step~\ref{step:add_u_to_level3}~$O((p^2 \deg(u) (p+\log p))$, and
	Step~\ref{step:add_direct_Ts} and~\ref{step:add_indirect_Ts} together~$O(p^4(p+\log p))$. We can summarize the running time of these construction steps
	as~$O(p^5 + p^3 \deg_G(u))$.
	
	To find the augmenting path, we construct an auxiliary flow network~$\Pi^\star$
	by splitting every vertex~$x \in \hat\Pi$ with unit capacity into~$x^-,x+$; add the
	arc~$x^-x+$ and then change very arc~$xy$ to~$x^+y^-$. This construction takes time~$|E(\hat \Pi)| = O(p^2)$. We then modify~$\Pi^\star$ to be the residual network for the
	flow~$f_H$ by taking each~$x \in V(H)$ and inverting the arc~$x^-x+$, as well
	as every~$xy \in E(H)$ and inverting the arc~$x^+y^-$, both in time~$O(p)$.
	Finally, we find a shortest path from the root~$v$ to~$T_{2,3}$ in $\Pi^\star$, using BFS, which takes time~$O(|E(\Pi^\star)| + |V(\Pi^\star)|) = O(p^2)$.
	Modifying the packing with the result augmenting path, if it exists, is easily subsumed by this running time.
    \qed
\end{proof}

Before we prove the main theorem, we will need to bound how often a packing
for a vertex~$u \in L$ needs to be updated, in particular with the more expensive update operations. 

\begin{replemma}{lemma:DecreasingPotential}
	Let~$L,R$ of~$V(G)$ be partition of~$G$ such that~$\pp_L^3(x) \leq p$
	for all~$x \in R$, let~$v \in R$. For every vertex~$u \in L$ with
	$T_1 \defeq \Target^1_L(u)$, $T_2 \defeq \Target^2_L(u) \setminus T_1$,
	and~$T_3 \defeq \Target^3_L(u) \setminus (T_1 \cup T_2)$ define the
	potential function~$\phi_L(u) \defeq |T_1|p^2 + |T_2|p + |T_3|$.

	Then for all~$u \in L$ if~$v \in \Target^3_{L\cup \{v\}}(u)$ then $\phi_{L \cup \{v\}}(u) > \phi_L(u)$ and otherwise $\phi_{L \cup \{v\}}(u) = \phi_L(u)$.
\end{replemma}
\begin{proof}
	It is easy to see that the sets~$T_1, T_2, T_3$ for~$u$ do not change
	if~$v \not \in \Target^3_{L\cup \{v\}}(u)$ and therefore the potential 
	function remains unchanged. Assume therefore that~$v \in \Target^3_{L\cup \{v\}}(u)$.

	To that end, let~$T_1, T_2, T_3$ describe the three target sets for~$u$
	under the partition~$L \cup \{v\}, R\setminus \{v\}$ and $T'_1, T'_2,T'_3$
	under the partition~$L, R$. Let us consider the potential difference
	\begin{align*}
			\Delta &:= \phi_{L\cup\{v\}}(u) - \phi_L(u) \\
			 		&= |T_1|p^2 + |T_2|p + |T_3| - (|T'_1|p^2 + |T'_2|p + |T'_3|) \\
			 		&= (|T_1|-|T'_1|)p^2 + (|T_2| - |T'_2|)p + |T_3| - |T'_3|
	\end{align*}
	Showing that~$\Delta > 0$ in all cases proves that~$\phi_{L\cup\{v\}}(u) > \phi_L(u)$. Note that~$v \in T_1 \cup T_2 \cup T_3$,
	we will now consider the different positions of~$v$ in these three sets
	and argue about how many new vertices could be added to them due to moving~$v$ across the partition.

	Consider first the case that~$v \in T_3$. Then no $(3,L)$-admissible path from~$u$ can contain~$v$, and we conclude that~$T'_1 = T_1$, $T'_2 = T_2$,
	and~$T'_3 = T_3 \setminus \{v\}$. It follows that~$\Delta = 1$.

	Next, consider the case that~$v \in T_2$. Then any~$(3,L)$-admissible path that contains~$v$ must have the shape~$uxvy$, with~$x \in R$ and~$y \in L$,
	in other words, $y \in N_L(v)$. Note that~$u \not \in N_L(v)$ as otherwise
	$v \not \in T_1$. This means that we can construct a $(3,L)$-admissible path-packing rooted at~$v$ using~$N_L(v)$ and a suitable~$(2,L)$-admissible
	path from~$u$ to~$v$, resulting in a path-packing of size~$|N_L(v)| + 1$.
	Since~$\pp^3_L(u) \leq p$, we have that $|N_L(v)| \leq \pp^3_L(u) - 1 \leq p -1$, and we conclude that $T'_1 = T_1$, $T'_2 = T_2 \setminus \{v\}$,
	and~$|T'_3| \leq |T_3| + (p-1)$. Therefore $\Delta \geq 1p - p + 1 = 1$.

	Finally, consider the case that~$v \in T_1$. Then any shortest $(3,L)$-admissible path that contains~$v$ must have the shape~$uvy$ or~$uvxz$, where~$y \in N_L(v) \setminus \{u\}$, $x \in R$ and~$z \in \Target^2_L(v)\setminus N_L(v)$.
	Collect the vertices of the first kind in a set~$Z_2$ and vertices
	of the second kind in a set~$Z_3$, note in particular that~$u \not \in Z_1 \cup Z_2$. Then~$T'_1 = T_1 \setminus\{v\}$,
	$T'_2 = T_2 \cup Z_2$ and~$T'_3 = T_3 \cup Z_3$ and
	$\Delta = p^2 - |Z_2|p - |Z_3|$, so we are left arguing that
	$|Z_2|p + |Z_3| < p^2$.

	Construct a tree~$\Gamma$ of~$(2,L)$-admissible paths from each vertex in~$Z_3$ to~$v$ and let~$X_3 \defeq N_\Gamma(v)$ be the `intermediate' vertices between~$v$ and~$Z_3$. Since~$X_3 \subseteq R$,
	for each~$x \in X_3$ it holds that~$\pp^3_L(x) \leq p$. Since~$u \not \in Z_3$, note that for each~$x$ we can construct a $(3,L)$-admissible path packing using the children of~$x$ in~$\Gamma$ as well as the path~$xvu$.
	Therefore, each vertex~$x$ has at most~$p-1$ children in~$\Gamma$ and
	thus~$|Z_3| \leq |X_3|(p-1)$. 

	Note further that we can construct a $(3,L)$-admissible packing rooted
	at~$v$ by using~$Z_2$ and~$v$ as direct neighbours, as well as~$|X_3|$
	paths into~$Z_3$, therefore~$|Z_2| + |X_3| \leq p - 1$. Therefore
	\begin{align*}
		|Z_2|p + |Z_3| &\leq |Z_2|p + |X_3|(p-1) \\
					&\leq (|Z_2|+|X_3|)p \leq (p-1)p <p^2,
	\end{align*}
	and we conclude that~$\Delta \geq 1$.

	Therefore in all three cases it holds that~$\Delta \geq 1$ and therefore
	that $\phi_{L \cup \{v\}}(u) > \phi_L(u)$, as claimed.
    \qed
\end{proof}

\begin{replemma}{lemma:PotentialBound}
	Assume that we reach a point in the algorithm where~$\Pack(u)$, $u \in L$, has size~$p$ for the first time after the call to \ref{alg:SimpleUpdate}. 
	Then~$\phi_L(u) \leq p^3$.
\end{replemma}
\begin{proof}
	Let~$T_1, T_2, T_3$ be defined as in Lemma~\ref{lemma:DecreasingPotential}
	for~$u$. Let~$H = \Pack(u)$ by Lemma~\ref{lemma:SimpleUpdate} we have that~$H$ is chordless and covering, in particular all vertices $T_1$ appear in~$N_H(u)$. Let~$S_1 = N_H(u) \setminus T_1$.
	
	By Lemma~\ref{lemma:PackingProperties}, for every~$y \in T_2 \cup T_3$, there exists either a path~$uwy$
	or~$uwxy$ with~$x \in R$ and~$w \in H$. For each such vertex~$w$, we have
	that~$\Target^2_L(G) \leq (p-1)^2$ by the usual tree argument we have used several times already. Thus~$|T_2| + |T_3| \leq |S_1|(p-1)^2$.

	Putting these bounds together, we have that
	\begin{align*}
		\phi_L(u) &= |T_1|p^2 + |T_2|p + |T_3| \\
							&\leq |T_1|p^2 + (|T_2| + |T_3|) \\
							&\leq |T_1|p^2 + |S_1|(p-1)^2 \leq (|T_1| + |S_1|)p^2 \\
							&= \Pack(u) \cdot p^2 = p^3. 
	\end{align*}
    \qed
\end{proof}

\repeattheorem{thmMain}
\begin{proof}
	By Lemma~\ref{lemma:Stage2Update}, a vertex~$u$ is added to~$\Candidates$
	in \ref{alg:Query} if~$\Pack(u)$ is a maximal $(3,L)$-admissible packing
	of size~$p$, proving that~$\pp^3_L(u) \leq p$. Therefore every vertex~$v$
	returned by the Oracle satisfies~$\pp^3_L(v) \leq p$ for the current set~$L$.
	As the path-packing number of~$v$ can only decrease when further vertices are moved from~$L$ to~$R$ (\cf Lemma~\ref{lemma:Candidates}) then the ordering~$\G$
	returned by the \ref{alg:Main}, assuming that the Oracle never returns FALSE, indeed satisfies~$\adm_3(\G) \leq p$ and therefore~$\adm_3(G) \leq p$.

	If, on the other hand, the Oracle returns FALSE, by Lemma~\ref{alg:Stage2Update}, for all~$v \in L$ it holds that~$\Pack(v)$ is a $(3,L)$-admissible packing of size
	at least~$p+1$. By Lemma~\ref{lemma:Greedy}, this means that~$G$ has~$\adm_3(G) > p$ and the algorithm returns FALSE.

	We showed in Lemma~\ref{lemma:Vias} that the maintenance of~$\Vias$ takes to~$O(p^2m)$ time and $O(p^3n)$ space over the whole algorithm run. Maintaining the two-packings costs, by Lemma~\ref{lemma:2Packings}, $O(mp \log p)$ time and~$O(pn)$ space.
		
	To bound the cost of updating the packings for vertices in~$L$, let us
	first observe that \ref{alg:SimpleUpdate} is called for at most~$|\Target^3_L(v)| \leq p^3$ vertices at a cost of $O(p^3 \log p)$, therefore the total cost of these calls over the whole run is bounded by~$O(n p^6 \log p)$.

	We then need to bound how often~\ref{alg:Stage1Update} and \ref{alg:Stage2Update} could be called for the same vertex~$u$. By Lemma~\ref{lemma:PotentialBound}
	the potential for~$u$ at this point is~$\phi_L(u) \leq p^3$ and by
	Lemma~\ref{lemma:DecreasingPotential} this potential decreases by at least one every time~$|\Pack(u)|$ drops down to~$p$. If~$\phi_L(u) < p$, then 
	$\Target^3_L(u) < p$ and~$u$ cannot have a packing of size~$p$ any more.
	We conclude that~\ref{alg:Stage1Update} and \ref{alg:Stage2Update}
	arc called at most~$p^3$ times per vertex~$u$ with costs
	$O(\deg(u) p^4)$ (Lemma~\ref{lemma:Stage1Update}) and $O(p^5 + p^3 \deg(u))$ (Lemma~\ref{lemma:Stage2Update}), respectively. The cost of these calls over the whole run is therefore bounded by~$O(n p^6 + m p^7) = O(m p^7)$.
\end{proof}

\newpage

\onecolumn
\subsection*{Experimental results}
The following table contains the complete experimental results. We abbreviated some network names for the sake of space.

\newcommand{\abbr}[1]{\textcolor{gray}{#1}}
\begingroup
\footnotesize
\input{notebooks/table}
\endgroup
\end{document}

%% file: notebooks/table.tex
\begin{longtable}{@{}l@{}rrrrrrrrrr@{}}
\toprule
 &  &  &  &  &  &  &  & \multicolumn{1}{c}{Time} & \multicolumn{1}{c}{Peak} & \multicolumn{1}{c}{Network} \\
Network & $\adm_2$ & $\adm_3$ & $\bar d$ & $\text{deg}$ & $\Delta$ & $m$ & $n$ & (seconds) & mem. (mB) & mem. (mB) \\
\midrule
\endhead
\bottomrule
\endfoot
AS-oregon-1 & 28 & 35 & 4.19 & 17 & 2389 & 23409 & 11174 & 16.09 & 39.81 & 1.52 \\
AS-oregon-2 & 52 & 62 & 5.71 & 31 & 2432 & 32730 & 11461 & 19.92 & 38.69 & 1.60 \\
\abbr{BG}-\abbr{AC}-\abbr{Lumin.} & 8 & 9 & 2.51 & 6 & 376 & 2312 & 1840 & 0.18 & 2.45 & 0.43 \\
\abbr{BG}-\abbr{AC}-Ms & 183 & 202 & 15.90 & 58 & 2217 & 321887 & 40495 & 1063.41 & 1338.17 & 8.96 \\
\abbr{BG}-\abbr{AC}-Rna & 75 & 75 & 6.22 & 54 & 3572 & 42815 & 13765 & 38.30 & 185.05 & 1.79 \\
\abbr{BG}-\abbr{AC}-Western & 64 & 82 & 6.09 & 17 & 535 & 64046 & 21028 & 45.58 & 146.40 & 3.08 \\
\abbr{BG}-All & 476 & 476 & 34.86 & 134 & 3620 & 1316843 & 75550 & 13541.29 & 8311.50 & 28.27 \\
\abbr{BG}-\abbr{A.}-Thaliana-Columbia & 53 & 68 & 9.20 & 26 & 1341 & 47916 & 10417 & 36.97 & 99.03 & 1.80 \\
\abbr{BG}-Biochemical-Activity & 29 & 36 & 4.12 & 11 & 427 & 17746 & 8620 & 4.37 & 25.57 & 1.54 \\
\abbr{BG}-Bos-Taurus & 4 & 4 & 1.87 & 3 & 27 & 424 & 454 & 0.05 & 0.44 & 0.17 \\
\abbr{BG}-\abbr{C.}-Elegans & 70 & 73 & 7.40 & 64 & 522 & 23646 & 6394 & 8.95 & 46.52 & 0.96 \\
\abbr{BG}-\abbr{C.}-Albicans-Sc5314 & 9 & 9 & 2.87 & 9 & 427 & 1609 & 1121 & 0.12 & 1.31 & 0.26 \\
\abbr{BG}-Canis-Familiaris & 2 & 2 & 1.75 & 2 & 90 & 125 & 143 & 0.02 & 0.10 & 0.10 \\
\abbr{BG}-Chemicals & 1 & 1 & 1.69 & 1 & 413 & 28093 & 33266 & 0.42 & 28.22 & 5.32 \\
\abbr{BG}-Co-Crystal-Structure & 5 & 5 & 1.76 & 5 & 92 & 2021 & 2291 & 0.07 & 1.92 & 0.42 \\
\abbr{BG}-Co-Fractionation & 83 & 83 & 10.23 & 83 & 187 & 56354 & 11017 & 27.35 & 112.47 & 1.93 \\
\abbr{BG}-Co-Localization & 9 & 13 & 2.51 & 6 & 63 & 4452 & 3543 & 0.21 & 4.20 & 0.43 \\
\abbr{BG}-Co-Purification & 12 & 12 & 2.76 & 12 & 1972 & 5970 & 4326 & 1.42 & 6.60 & 0.79 \\
\abbr{BG}-Cricetulus-Griseus & 1 & 1 & 1.65 & 1 & 30 & 57 & 69 & 0.03 & 0.09 & 0.09 \\
\abbr{BG}-Danio-Rerio & 3 & 3 & 2.04 & 3 & 61 & 266 & 261 & 0.02 & 0.22 & 0.12 \\
\abbr{BG}-\abbr{D.}-Discoideum-Ax4 & 1 & 1 & 1.48 & 1 & 4 & 20 & 27 & 0.01 & 0.08 & 0.08 \\
\abbr{BG}-Dosage-Growth-Defect & 9 & 10 & 3.03 & 5 & 213 & 2193 & 1447 & 0.11 & 1.43 & 0.26 \\
\abbr{BG}-Dosage-Lethality & 8 & 9 & 2.58 & 4 & 392 & 2289 & 1776 & 0.23 & 2.06 & 0.26 \\
\abbr{BG}-Dosage-Rescue & 11 & 18 & 3.81 & 7 & 75 & 6444 & 3380 & 0.42 & 6.65 & 0.44 \\
\abbr{BG}-\abbr{D.}-Melanogaster & 83 & 104 & 12.98 & 83 & 303 & 60556 & 9330 & 60.16 & 211.36 & 1.96 \\
\abbr{BG}-\abbr{E.}-Nidulans-Fgsc-A4 & 2 & 2 & 1.94 & 2 & 44 & 62 & 64 & 0.03 & 0.09 & 0.09 \\
\abbr{BG}-\abbr{E.}-Coli-K12-Mg1655 & 10 & 13 & 2.97 & 5 & 58 & 1889 & 1273 & 0.09 & 1.85 & 0.26 \\
\abbr{BG}-\abbr{E.}-Coli-K12-W3110 & 290 & 305 & 89.40 & 133 & 1187 & 181620 & 4063 & 1259.80 & 824.25 & 3.77 \\
\abbr{BG}-Far-Western & 3 & 3 & 1.82 & 3 & 60 & 1089 & 1199 & 0.03 & 1.04 & 0.25 \\
\abbr{BG}-Fret & 24 & 24 & 2.82 & 19 & 51 & 2395 & 1700 & 0.07 & 1.55 & 0.25 \\
\abbr{BG}-Gallus-Gallus & 4 & 5 & 2.11 & 4 & 110 & 436 & 413 & 0.03 & 0.33 & 0.12 \\
\abbr{BG}-Glycine-Max & 2 & 2 & 1.77 & 2 & 13 & 39 & 44 & 0.02 & 0.09 & 0.09 \\
\abbr{BG}-Hepatitus-C-Virus & 1 & 1 & 1.97 & 1 & 133 & 134 & 136 & 0.02 & 0.10 & 0.10 \\
\abbr{BG}-Homo-Sapiens & 263 & 280 & 30.69 & 71 & 2882 & 369767 & 24093 & 2070.58 & 1867.87 & 9.40 \\
\abbr{BG}-\abbr{HHV}-1 & 3 & 3 & 2.34 & 3 & 40 & 208 & 178 & 0.02 & 0.14 & 0.10 \\
\abbr{BG}-\abbr{HHV}-4 & 2 & 2 & 2.02 & 2 & 154 & 326 & 323 & 0.02 & 0.19 & 0.12 \\
\abbr{BG}-\abbr{HHV}-5 & 1 & 1 & 1.77 & 1 & 27 & 107 & 121 & 0.01 & 0.10 & 0.10 \\
\abbr{BG}-\abbr{HHV}-8 & 3 & 3 & 1.93 & 3 & 119 & 691 & 716 & 0.03 & 0.58 & 0.17 \\
\abbr{BG}-\abbr{HIV}-1 & 6 & 7 & 2.32 & 3 & 324 & 1319 & 1138 & 0.11 & 1.30 & 0.26 \\
\abbr{BG}-\abbr{HIV}-2 & 1 & 1 & 1.58 & 1 & 6 & 15 & 19 & 0.03 & 0.08 & 0.08 \\
\abbr{BG}-\abbr{HPV}-16 & 2 & 2 & 2.15 & 2 & 93 & 186 & 173 & 0.02 & 0.12 & 0.10 \\
Cannes2013 & 114 & 167 & 3.82 & 27 & 15169 & 835892 & 438089 & 6893.26 & 3244.63 & 47.83 \\
CoW-interstate & 7 & 7 & 3.51 & 4 & 25 & 319 & 182 & 0.02 & 0.23 & 0.10 \\
DNC-emails & 28 & 29 & 4.70 & 17 & 402 & 4384 & 1866 & 0.54 & 5.30 & 0.46 \\
EU-email-core & 74 & 81 & 32.58 & 34 & 345 & 16064 & 986 & 7.75 & 29.49 & 0.44 \\
JDK\_dependency & 76 & 78 & 16.68 & 65 & 5923 & 53658 & 6434 & 70.38 & 140.39 & 1.38 \\
JUNG-javax & 76 & 78 & 16.43 & 65 & 5655 & 50290 & 6120 & 66.08 & 133.60 & 1.32 \\
NYClimateMarch2014 & 161 & 190 & 6.39 & 34 & 14687 & 327080 & 102378 & 3135.85 & 1401.96 & 13.59 \\
NZ\_legal & 68 & 75 & 14.70 & 25 & 429 & 15739 & 2141 & 9.30 & 33.57 & 0.55 \\
Noordin-terror-loc & 4 & 4 & 2.99 & 3 & 18 & 190 & 127 & 0.02 & 0.16 & 0.10 \\
Noordin-terror-orgas & 3 & 4 & 2.81 & 3 & 21 & 181 & 129 & 0.02 & 0.15 & 0.10 \\
Noordin-terror-relation & 11 & 11 & 7.17 & 11 & 28 & 251 & 70 & 0.02 & 0.12 & 0.09 \\
ODLIS & 38 & 50 & 11.29 & 12 & 592 & 16377 & 2900 & 10.62 & 30.06 & 0.58 \\
Opsahl-forum & 42 & 46 & 15.65 & 14 & 128 & 7036 & 899 & 1.55 & 13.13 & 0.34 \\
Opsahl-socnet & 61 & 67 & 14.57 & 20 & 255 & 13838 & 1899 & 5.60 & 29.59 & 0.59 \\
StackOverflow-tags & 6 & 6 & 4.26 & 6 & 16 & 245 & 115 & 0.02 & 0.17 & 0.10 \\
Y2H\_union & 7 & 10 & 2.75 & 4 & 89 & 2705 & 1966 & 0.15 & 2.84 & 0.41 \\
Yeast & 18 & 27 & 6.08 & 6 & 66 & 7182 & 2361 & 1.02 & 10.00 & 0.47 \\
actor\_movies & 105 & 112 & 5.75 & 14 & 646 & 1470404 & 511463 & 2708.59 & 4011.90 & 101.29 \\
advogato & 86 & 95 & 15.24 & 25 & 803 & 39285 & 5155 & 38.12 & 94.60 & 1.17 \\
airlines & 18 & 20 & 11.04 & 13 & 130 & 1297 & 235 & 0.09 & 0.86 & 0.14 \\
american\_revolution & 3 & 3 & 2.27 & 3 & 59 & 160 & 141 & 0.02 & 0.12 & 0.10 \\
as-22july06 & 44 & 52 & 4.22 & 25 & 2390 & 48436 & 22963 & 52.21 & 110.27 & 3.05 \\
as20000102 & 21 & 25 & 3.88 & 12 & 1458 & 12572 & 6474 & 3.94 & 16.94 & 0.81 \\
autobahn & 3 & 3 & 2.56 & 2 & 5 & 478 & 374 & 0.05 & 0.35 & 0.12 \\
bahamas & 8 & 10 & 2.24 & 6 & 14902 & 246291 & 219856 & 316.10 & 299.92 & 21.84 \\
bergen & 12 & 12 & 10.26 & 9 & 32 & 272 & 53 & 0.02 & 0.12 & 0.09 \\
bitcoin-otc-negative & 21 & 22 & 4.06 & 16 & 227 & 3259 & 1606 & 0.29 & 3.69 & 0.26 \\
bitcoin-otc-positive & 50 & 60 & 6.67 & 20 & 788 & 18591 & 5573 & 11.56 & 38.79 & 0.87 \\
bn-fly-\abbr{d.}\_medulla\_1 & 44 & 51 & 10.01 & 18 & 927 & 8911 & 1781 & 2.68 & 16.32 & 0.35 \\
bn-mouse\_retina\_1 & 223 & 237 & 168.79 & 121 & 744 & 90811 & 1076 & 100.52 & 87.10 & 1.91 \\
boards\_gender\_1m & 25 & 25 & 9.67 & 25 & 88 & 19993 & 4134 & 1.37 & 17.09 & 0.93 \\
boards\_gender\_2m & 7 & 10 & 2.65 & 4 & 45 & 5598 & 4220 & 0.25 & 6.81 & 0.79 \\
ca-CondMat & 30 & 51 & 8.08 & 25 & 279 & 93439 & 23133 & 29.62 & 146.34 & 3.45 \\
ca-GrQc & 43 & 43 & 5.53 & 43 & 81 & 14484 & 5241 & 0.74 & 12.17 & 0.85 \\
ca-HepPh & 238 & 238 & 19.74 & 135 & 491 & 118489 & 12006 & 293.45 & 275.42 & 2.93 \\
capitalist & 21 & 23 & 15.41 & 19 & 91 & 1071 & 139 & 0.12 & 0.68 & 0.11 \\
celegans & 21 & 24 & 14.46 & 10 & 134 & 2148 & 297 & 0.25 & 2.27 & 0.14 \\
chess & 88 & 102 & 15.31 & 29 & 181 & 55899 & 7301 & 44.74 & 172.76 & 2.16 \\
chicago & 1 & 1 & 1.77 & 1 & 12 & 1298 & 1467 & 0.03 & 1.14 & 0.25 \\
cit-HepPh & 89 & 140 & 24.37 & 30 & 846 & 420877 & 34546 & 940.86 & 1467.96 & 10.25 \\
cit-HepTh & 128 & 178 & 25.37 & 37 & 2468 & 352285 & 27769 & 943.92 & 1261.99 & 7.52 \\
codeminer & 5 & 6 & 2.80 & 4 & 55 & 1015 & 724 & 0.03 & 0.83 & 0.17 \\
columbia-mobility & 9 & 11 & 9.61 & 9 & 228 & 4147 & 863 & 0.29 & 2.98 & 0.21 \\
columbia-social & 19 & 20 & 17.90 & 18 & 545 & 7724 & 863 & 0.81 & 5.95 & 0.26 \\
cora\_citation & 30 & 48 & 7.70 & 13 & 377 & 89157 & 23166 & 26.98 & 127.84 & 3.40 \\
countries & 16 & 17 & 2.11 & 6 & 110602 & 624402 & 592414 & 13906.07 & 795.14 & 89.88 \\
cpan-authors & 17 & 18 & 5.03 & 9 & 327 & 2112 & 839 & 0.21 & 0.94 & 0.18 \\
deezer & 60 & 108 & 18.26 & 21 & 420 & 498202 & 54573 & 621.39 & 1562.55 & 11.76 \\
digg & 46 & 79 & 5.68 & 8 & 285 & 86312 & 30398 & 88.10 & 249.38 & 6.34 \\
diseasome & 11 & 11 & 3.86 & 11 & 84 & 2738 & 1419 & 0.08 & 2.04 & 0.26 \\
dolphins & 6 & 7 & 5.13 & 4 & 12 & 159 & 62 & 0.02 & 0.12 & 0.09 \\
dutch-textiles & 5 & 5 & 3.75 & 5 & 31 & 90 & 48 & 0.02 & 0.09 & 0.09 \\
ecoli-transcript & 5 & 5 & 2.73 & 3 & 74 & 578 & 423 & 0.03 & 0.38 & 0.12 \\
edinburgh\_\abbr{assoc.}\_thesaurus & 197 & 203 & 25.69 & 34 & 1062 & 297094 & 23132 & 1303.58 & 1931.91 & 6.62 \\
email-Enron & 145 & 169 & 10.02 & 43 & 1383 & 183831 & 36692 & 547.25 & 532.66 & 7.30 \\
escorts & 45 & 52 & 4.67 & 11 & 305 & 39044 & 16730 & 19.23 & 81.09 & 3.13 \\
euroroad & 3 & 3 & 2.41 & 2 & 10 & 1417 & 1174 & 0.04 & 1.20 & 0.25 \\
eva-corporate & 4 & 4 & 1.85 & 3 & 552 & 6711 & 7253 & 0.26 & 8.01 & 1.40 \\
exnet-water & 3 & 3 & 2.55 & 2 & 10 & 2416 & 1893 & 0.04 & 2.05 & 0.42 \\
facebook-links & 191 & 226 & 25.64 & 52 & 1098 & 817090 & 63731 & 3418.32 & 4394.95 & 20.10 \\
foldoc & 36 & 69 & 13.70 & 12 & 728 & 91471 & 13356 & 53.65 & 160.98 & 2.45 \\
foodweb-caribbean & 23 & 26 & 13.47 & 13 & 196 & 3313 & 492 & 0.26 & 1.52 & 0.20 \\
foodweb-otago & 23 & 23 & 11.80 & 14 & 45 & 832 & 141 & 0.07 & 0.55 & 0.11 \\
football & 11 & 11 & 10.66 & 8 & 12 & 613 & 115 & 0.02 & 0.54 & 0.10 \\
google+ & 38 & 42 & 3.32 & 12 & 2761 & 39194 & 23628 & 21.18 & 58.28 & 3.03 \\
gowalla & 202 & 251 & 9.67 & 51 & 14730 & 950327 & 196591 & 9271.88 & 4049.49 & 30.90 \\
haggle & 40 & 40 & 15.50 & 39 & 101 & 2124 & 274 & 0.23 & 2.09 & 0.14 \\
hex & 4 & 5 & 5.62 & 3 & 6 & 930 & 331 & 0.03 & 0.65 & 0.13 \\
hypertext\_2009 & 43 & 43 & 38.87 & 28 & 98 & 2196 & 113 & 0.16 & 0.88 & 0.13 \\
ia-email-univ & 21 & 29 & 9.62 & 11 & 71 & 5451 & 1133 & 0.73 & 8.74 & 0.29 \\
ia-infect-dublin & 21 & 22 & 13.49 & 17 & 50 & 2765 & 410 & 0.22 & 2.48 & 0.15 \\
ia-reality & 12 & 16 & 2.26 & 5 & 261 & 7680 & 6809 & 0.70 & 10.16 & 0.78 \\
infectious & 21 & 22 & 13.49 & 17 & 50 & 2765 & 410 & 0.19 & 2.43 & 0.15 \\
ingredients & 475 & 476 & 197.46 & 136 & 3426 & 431654 & 4372 & 1768.98 & 683.14 & 9.99 \\
iscas89-s1196 & 4 & 5 & 2.85 & 2 & 16 & 537 & 377 & 0.06 & 0.43 & 0.12 \\
iscas89-s1238 & 5 & 5 & 3.00 & 2 & 18 & 625 & 416 & 0.03 & 0.48 & 0.12 \\
iscas89-s13207 & 6 & 6 & 2.73 & 4 & 37 & 3406 & 2492 & 0.07 & 2.57 & 0.43 \\
iscas89-s1423 & 3 & 3 & 2.62 & 2 & 17 & 554 & 423 & 0.02 & 0.37 & 0.12 \\
iscas89-s1488 & 7 & 7 & 3.37 & 3 & 53 & 779 & 463 & 0.04 & 0.61 & 0.17 \\
iscas89-s1494 & 7 & 7 & 3.37 & 3 & 56 & 796 & 473 & 0.04 & 0.67 & 0.17 \\
iscas89-s15850 & 4 & 5 & 2.47 & 4 & 25 & 4004 & 3247 & 0.08 & 2.84 & 0.41 \\
iscas89-s27 & 1 & 1 & 1.78 & 1 & 3 & 8 & 9 & 0.03 & 0.08 & 0.08 \\
iscas89-s298 & 3 & 3 & 2.85 & 2 & 11 & 131 & 92 & 0.02 & 0.09 & 0.09 \\
iscas89-s344 & 3 & 3 & 2.44 & 2 & 9 & 122 & 100 & 0.02 & 0.09 & 0.09 \\
iscas89-s349 & 3 & 3 & 2.49 & 2 & 9 & 127 & 102 & 0.02 & 0.09 & 0.09 \\
iscas89-s35932 & 2 & 2 & 2.55 & 2 & 1440 & 15961 & 12515 & 0.47 & 10.43 & 1.46 \\
iscas89-s382 & 4 & 4 & 2.90 & 2 & 18 & 168 & 116 & 0.03 & 0.13 & 0.10 \\
iscas89-s38417 & 6 & 6 & 2.24 & 4 & 39 & 10635 & 9500 & 0.27 & 9.63 & 1.45 \\
iscas89-s38584 & 7 & 7 & 2.74 & 4 & 54 & 12573 & 9193 & 0.46 & 10.90 & 1.47 \\
iscas89-s386 & 4 & 4 & 3.51 & 3 & 23 & 200 & 114 & 0.05 & 0.15 & 0.10 \\
iscas89-s400 & 4 & 4 & 3.01 & 2 & 19 & 182 & 121 & 0.02 & 0.14 & 0.10 \\
iscas89-s444 & 4 & 4 & 3.07 & 2 & 19 & 206 & 134 & 0.03 & 0.15 & 0.10 \\
iscas89-s510 & 4 & 6 & 2.92 & 2 & 12 & 251 & 172 & 0.02 & 0.19 & 0.10 \\
iscas89-s526 & 4 & 4 & 3.38 & 3 & 12 & 270 & 160 & 0.02 & 0.18 & 0.10 \\
iscas89-s526n & 4 & 4 & 3.37 & 3 & 12 & 268 & 159 & 0.02 & 0.19 & 0.10 \\
iscas89-s5378 & 5 & 5 & 2.32 & 3 & 10 & 1639 & 1411 & 0.06 & 1.25 & 0.25 \\
iscas89-s641 & 4 & 4 & 2.88 & 3 & 12 & 144 & 100 & 0.02 & 0.10 & 0.09 \\
iscas89-s713 & 4 & 4 & 2.63 & 3 & 12 & 180 & 137 & 0.02 & 0.14 & 0.10 \\
iscas89-s820 & 9 & 9 & 4.02 & 3 & 48 & 480 & 239 & 0.03 & 0.35 & 0.13 \\
iscas89-s832 & 9 & 9 & 4.07 & 3 & 49 & 498 & 245 & 0.03 & 0.36 & 0.13 \\
iscas89-s9234 & 4 & 4 & 2.39 & 4 & 18 & 2370 & 1985 & 0.05 & 2.02 & 0.41 \\
iscas89-s953 & 3 & 4 & 2.73 & 2 & 12 & 454 & 332 & 0.02 & 0.36 & 0.12 \\
jazz & 30 & 36 & 27.70 & 29 & 100 & 2742 & 198 & 0.30 & 2.09 & 0.14 \\
karate & 4 & 4 & 4.59 & 4 & 17 & 78 & 34 & 0.02 & 0.09 & 0.09 \\
lederberg & 47 & 64 & 9.98 & 15 & 1103 & 41532 & 8324 & 27.38 & 92.82 & 1.81 \\
lesmiserables & 9 & 9 & 6.60 & 9 & 36 & 254 & 77 & 0.02 & 0.14 & 0.09 \\
link-pedigree & 2 & 3 & 2.51 & 2 & 14 & 1125 & 898 & 0.03 & 0.97 & 0.25 \\
linux & 106 & 125 & 13.83 & 23 & 9338 & 213217 & 30834 & 1071.45 & 501.17 & 7.99 \\
loc-brightkite\_edges & 85 & 122 & 7.35 & 52 & 1134 & 214078 & 58228 & 509.98 & 700.84 & 12.74 \\
location & 16 & 16 & 2.61 & 5 & 12189 & 293697 & 225486 & 210.99 & 275.75 & 22.86 \\
marvel & 58 & 63 & 9.95 & 18 & 1625 & 96662 & 19428 & 98.00 & 108.74 & 3.47 \\
mg\_casino & 9 & 9 & 5.98 & 9 & 94 & 326 & 109 & 0.02 & 0.14 & 0.09 \\
mg\_forrestgump & 8 & 8 & 5.77 & 8 & 89 & 271 & 94 & 0.03 & 0.12 & 0.09 \\
mg\_godfatherII & 8 & 8 & 5.62 & 8 & 34 & 219 & 78 & 0.02 & 0.10 & 0.09 \\
mg\_watchmen & 7 & 7 & 5.29 & 7 & 33 & 201 & 76 & 0.02 & 0.11 & 0.09 \\
minnesota & 3 & 3 & 2.50 & 2 & 5 & 3303 & 2642 & 0.05 & 2.50 & 0.42 \\
moreno\_health & 12 & 16 & 8.24 & 7 & 27 & 10455 & 2539 & 0.73 & 12.36 & 0.47 \\
mousebrain & 141 & 141 & 151.07 & 111 & 205 & 16089 & 213 & 2.32 & 3.86 & 0.43 \\
movielens\_1m & 554 & 652 & 205.26 & 135 & 3428 & 1000209 & 9746 & 10439.14 & 5397.62 & 22.07 \\
movies & 5 & 6 & 3.80 & 3 & 19 & 192 & 101 & 0.02 & 0.13 & 0.09 \\
muenchen-bahn & 3 & 3 & 2.59 & 2 & 13 & 578 & 447 & 0.03 & 0.40 & 0.12 \\
munin & 3 & 3 & 2.11 & 3 & 66 & 1397 & 1324 & 0.03 & 1.12 & 0.25 \\
netscience & 19 & 19 & 3.75 & 19 & 34 & 2742 & 1461 & 0.06 & 1.61 & 0.26 \\
offshore & 20 & 22 & 3.63 & 13 & 37336 & 505965 & 278877 & 3710.32 & 432.27 & 44.91 \\
openflights & 52 & 57 & 10.67 & 28 & 242 & 15677 & 2939 & 7.54 & 30.28 & 0.57 \\
p2p-Gnutella04 & 23 & 35 & 7.35 & 7 & 103 & 39994 & 10876 & 9.27 & 70.26 & 1.65 \\
panama & 62 & 62 & 2.52 & 62 & 7015 & 702437 & 556686 & 565.08 & 1173.10 & 88.11 \\
paradise & 55 & 59 & 2.93 & 23 & 35359 & 794545 & 542102 & 2817.62 & 1735.40 & 90.70 \\
photoviz\_dynamic & 7 & 8 & 3.24 & 4 & 29 & 610 & 376 & 0.03 & 0.43 & 0.12 \\
pigs & 3 & 3 & 2.41 & 2 & 39 & 592 & 492 & 0.02 & 0.51 & 0.17 \\
polblogs & 72 & 82 & 27.31 & 36 & 351 & 16715 & 1224 & 11.51 & 36.35 & 0.46 \\
polbooks & 9 & 9 & 8.40 & 6 & 25 & 441 & 105 & 0.04 & 0.26 & 0.10 \\
pollination-carlinville & 53 & 54 & 20.34 & 18 & 157 & 15255 & 1500 & 4.49 & 30.70 & 0.44 \\
pollination-daphni & 26 & 29 & 7.36 & 9 & 124 & 2933 & 797 & 0.40 & 3.17 & 0.19 \\
pollination-tenerife & 6 & 6 & 3.79 & 4 & 17 & 129 & 68 & 0.05 & 0.11 & 0.09 \\
pollination-uk & 76 & 88 & 33.97 & 35 & 256 & 16712 & 984 & 9.11 & 22.05 & 0.45 \\
ratbrain & 78 & 83 & 91.57 & 67 & 497 & 23030 & 503 & 5.11 & 13.83 & 0.54 \\
reactome & 184 & 184 & 46.64 & 62 & 855 & 147547 & 6327 & 125.70 & 280.92 & 3.24 \\
residence\_hall & 21 & 25 & 16.95 & 11 & 56 & 1839 & 217 & 0.14 & 1.86 & 0.12 \\
rhesusbrain & 37 & 41 & 25.24 & 19 & 111 & 3054 & 242 & 0.40 & 2.99 & 0.16 \\
roget-thesaurus & 11 & 17 & 7.22 & 6 & 28 & 3648 & 1010 & 0.24 & 4.51 & 0.27 \\
seventh-graders & 16 & 16 & 17.24 & 13 & 28 & 250 & 29 & 0.02 & 0.09 & 0.09 \\
slashdot\_threads & 74 & 105 & 4.60 & 13 & 2915 & 117378 & 51083 & 269.86 & 490.46 & 6.18 \\
soc-Epinions1 & 268 & 286 & 10.69 & 67 & 3044 & 405740 & 75879 & 4560.60 & 3340.97 & 15.74 \\
soc-Slashdot0811 & 232 & 262 & 12.13 & 54 & 2539 & 469180 & 77360 & 5020.69 & 4088.47 & 15.66 \\
soc-advogato & 86 & 95 & 15.26 & 25 & 807 & 39432 & 5167 & 37.71 & 98.34 & 1.17 \\
soc-gplus & 38 & 42 & 3.32 & 12 & 2761 & 39194 & 23628 & 21.21 & 57.92 & 3.03 \\
soc-hamsterster & 51 & 62 & 13.71 & 24 & 273 & 16630 & 2426 & 6.37 & 28.85 & 0.59 \\
soc-wiki-Vote & 16 & 20 & 6.56 & 9 & 102 & 2914 & 889 & 0.26 & 3.51 & 0.19 \\
sp\_data\_school\_day\_2 & 57 & 61 & 46.55 & 33 & 88 & 5539 & 238 & 0.73 & 3.95 & 0.20 \\
teams & 127 & 127 & 2.92 & 9 & 2671 & 1366466 & 935591 & 7182.31 & 4367.93 & 175.88 \\
train\_bombing & 10 & 10 & 7.59 & 10 & 29 & 243 & 64 & 0.02 & 0.11 & 0.09 \\
twittercrawl & 237 & 268 & 84.70 & 132 & 1084 & 154824 & 3656 & 983.21 & 614.49 & 3.33 \\
ukroad & 3 & 3 & 2.53 & 3 & 5 & 15641 & 12378 & 0.21 & 11.04 & 1.36 \\
unicode\_languages & 7 & 8 & 2.89 & 4 & 141 & 1255 & 868 & 0.07 & 0.91 & 0.17 \\
wafa-ceos & 7 & 7 & 7.15 & 5 & 22 & 93 & 26 & 0.03 & 0.08 & 0.08 \\
wafa-eies & 27 & 27 & 28.98 & 24 & 44 & 652 & 45 & 0.02 & 0.13 & 0.10 \\
wafa-hightech & 13 & 14 & 15.14 & 12 & 20 & 159 & 21 & 0.02 & 0.09 & 0.09 \\
wafa-padgett & 3 & 4 & 3.60 & 3 & 8 & 27 & 15 & 0.03 & 0.08 & 0.08 \\
web-EPA & 16 & 25 & 4.17 & 6 & 175 & 8909 & 4271 & 1.28 & 15.28 & 0.78 \\
web-california & 26 & 33 & 5.17 & 11 & 199 & 15969 & 6175 & 3.43 & 22.31 & 0.84 \\
web-google & 17 & 17 & 4.27 & 17 & 59 & 2773 & 1299 & 0.10 & 1.98 & 0.26 \\
wiki-vote & 162 & 183 & 28.32 & 53 & 1065 & 100762 & 7115 & 233.81 & 352.30 & 2.17 \\
wikipedia-norm & 59 & 67 & 16.34 & 22 & 455 & 15372 & 1881 & 7.19 & 25.49 & 0.58 \\
win95pts & 3 & 3 & 2.26 & 2 & 9 & 112 & 99 & 0.02 & 0.09 & 0.09 \\
windsurfers & 15 & 16 & 15.63 & 11 & 31 & 336 & 43 & 0.06 & 0.12 & 0.09 \\
word\_adjacencies & 11 & 12 & 7.59 & 6 & 49 & 425 & 112 & 0.03 & 0.28 & 0.10 \\
zewail & 55 & 79 & 16.29 & 18 & 331 & 54182 & 6651 & 36.97 & 140.82 & 1.41 \\
\end{longtable}